\documentclass[11pt]{article}
\usepackage[utf8]{inputenc}
\usepackage[english]{babel}
\usepackage[utf8]{inputenc}
\usepackage{amsmath}
\usepackage{graphicx}
\usepackage{setspace}
\usepackage[left=2cm, right=3cm, top=1in,bottom=1in]{geometry}
\usepackage[colorinlistoftodos]{todonotes}
\usepackage{listings}
\usepackage{booktabs}
\usepackage{dsfont}
\usepackage{ latexsym }
\usepackage{ dsfont }
\usepackage{amsthm}
\usepackage{natbib}
\usepackage{amsfonts}
\usepackage{ amssymb }
\usepackage{tikz}
\usetikzlibrary{positioning,calc}
\usetikzlibrary{arrows.meta, shapes}

\usepackage[backref=page]{hyperref}
\hypersetup{	colorlinks=true,	linkcolor=black,	filecolor=magenta,	urlcolor=black,	allcolors=black,}

\newif\ifhidetext
\hidetextfalse 

\newcommand{\hidetext}[1]{%
  \ifhidetext
  \else
    #1
  \fi
}

\usepackage{subcaption}

\usepackage{amsthm}
\usepackage{ mathrsfs }
\usepackage{amsfonts}

\usepackage{bm}

\newcommand{\cm}[1]{\ignorespaces}

\usepackage{booktabs,multirow,tabularx}
\usepackage{graphics}
\usepackage{graphicx}

\usepackage{xcolor}
\definecolor{mypink}{RGB}{219, 48, 122}

\definecolor{mypurple}{RGB}{75,0,130}

	\newtheorem{theorem}{Theorem}
	\newtheorem{lemma}{Lemma} 
	\newtheorem{proposition}{Proposition} 
	\newtheorem{definition}{Definition}

	\newtheorem{corollary}{Corollary}
	\newtheorem{assumption}{Assumption}

\newcommand{\bR}{\mathbb{R}}
\newcommand{\bN}{\mathbb{N}}

\newcommand{\cC}{\mathcal{C}}
\newcommand{\cU}{\mathcal{U}}
\newcommand{\cE}{\mathcal{E}}

\newcommand{\bE}{\mathbb{E}}

\newcommand{\bP}{P}
\newcommand{\bQ}{Q}

\newcommand{\bM}{\mathbb{M}}

\newcommand{\bG}{\mathbb{G}}

\newcommand{\supp}{\text{supp}}

\newcommand{\LHS}{\text{LHS}}

\newcommand{\br}[1]{\left\{ #1 \right\} }

\newcommand{\sbr}[1]{\left( #1 \right) }

\newcommand{\nbr}[1]{\left\| #1 \right\|}
\newcommand{\absbr}[1]{\left| #1 \right|}

\newcommand\numberthis{\addtocounter{equation}{1}\tag{\theequation}}

\def\dd{\text{d}}

\def\cX{\mathcal{X}} 
\def\cY{\mathcal{Y}} 

\def\leb{\lambda} 

\def\bftheta{\boldsymbol \theta}

\def\bfone{\boldsymbol 1}

\def\cN{\mathcal N}

\def\cG{\mathcal G}
\def\cP{\mathcal P}

\def\cF{\mathcal F}

\def\iid{\scriptsize \mbox{iid}}

\def\as{\scriptsize \mbox{a.s.}}

\def\rT{\mathrm T}
\def\R{\mathbb R}

\def\spacingset#1{\renewcommand{\baselinestretch}%
{#1}\small\normalsize} \spacingset{1}

\title{Distributional Evaluation of Generative Models via Relative Density Ratio}
\author{Yuliang Xu $^{1}$, Yun Wei $^2$, and Li Ma $^{1}$\\[4mm]
Department of Statistics \& Data Science Institute, University of Chicago$^1$\\
Department of Mathematical Science, University of Texas at Dallas$^2$}
\date{}

\begin{document}

\maketitle

\begin{abstract}
    We propose a function-valued evaluation metric for generative models based on the relative density ratio (RDR) designed to characterize distributional differences between real and generated samples. As an evaluation metric, the RDR function preserves $\phi$-divergence between two distributions, enables sample-level evaluation that facilitates downstream investigations of feature-specific distributional differences, and has a bounded range that affords clear interpretability and numerical stability. Function estimation of the RDR is achieved efficiently through optimization on the variational form of $\phi$-divergence. We provide theoretical convergence rate guarantees for general estimators based on M-estimator theory, as well as the convergence rate of neural network-based estimators when the true ratio is in the anisotropic Besov space. We demonstrate the power of the proposed RDR-based evaluation through numerical experiments on MNIST, CelebA64, and the American Gut project microbiome data. We show that the estimated RDR enables not only effective overall comparison of competing generative models, but also a convenient way to reveal the underlying nature of goodness-of-fit. This enables one to assess support overlap, coverage, and fidelity while pinpointing regions of the sample space where generators concentrate and revealing the features that drive the most salient distributional differences.
\end{abstract}

\noindent%
\vfill

\spacingset{1.2} 

\newpage

\section{Introduction}\label{sec:intro}

Generative AI has become an integral part of modern society, facilitating tasks in everyday life and advancing scientific discovery. With the remarkable development and deployment of generative models in domains ranging from text and image generation to scientific data synthesis, there is a growing interest in determining which models perform better in different contexts. 
We have witnessed great successes in text and image generation using various model classes, including VAE (variational autoencoders), GANs (generative adversarial networks), diffusion models, and continuous flow matching models. However, despite the realism of generated samples to the human eye, studies have revealed distributional shifts between samples produced by advanced generators and those observed in real data. This phenomenon has been observed for many modern generators, such as  GANs \citep{richardson2018gans} and diffusion models \citep{sehwag2022generating}, which are not easily detected by visual inspection. Such discrepancies are particularly concerning when generative AI is used for scientific applications, where the fidelity of the underlying data distribution is critical. Thus, a key step toward trustworthy generative modeling is to understand, at a distributional level, how and where the generator diverges from the true data generation mechanism.

Existing automated evaluation criteria for modern generative models have largely focused on benchmarking generators through {\em scalar} metrics. Popular metrics include the Fréchet Inception Distance \citep[FID]{heusel2017gans}  and the Inception Score \citep[IS]{salimans2016improved}, both of which rely on strong parametric assumptions or low-dimensional feature mappings to assess model performance on large-scale, high-dimensional data. \citet{pmlr-v97-simon19a} and \citet{sajjadi2018assessing} provide a precision-recall curve to describe the coverage and fidelity of generators based on low-dimensional features of the original data. While these metrics capture human-perceivable fidelity, they largely overlook underlying distributional discrepancies. As a result, even a generator that samples extensively from a narrow region of the true data support can achieve a high score, provided that the generated samples \textit{look} realistic. More importantly, in scientific data generation, as we will demonstrate in a microbiome example, visually convincing generations may be biased and fail to represent the true data-generating mechanism in important ways that could bias the downstream scientific analysis. This issue is critical in many scientific application domains where the generated samples are used to augment observed data \citep{choi2023deepmicrogen} and to perform downstream analyzes such as biomarker detection \citep{marouf2020realistic} or cell clustering \citep{lall2022lsh}. The goal of this work is to propose a {\em function-valued} metric, rather than a single scalar metric, capable of providing both distribution-level and sample-level assessments. Our approach is based on the relative density ratio, which is simple to implement and scalable to massive, high-dimensional datasets. We illustrate how this framework allows for the identification of features associated with the detected distributional discrepancies between observed and generated data.

Comparing the generative model distribution and the real data distribution can be viewed as the classic two-sample problem from a statistical perspective. The classic two-sample comparison methods in statistics have mostly adopted a hypothesis testing perspective. The predominant approach involves computing a test statistic, which is usually a scalar distance metric, and based on that, computing the p-value under its (asymptotic) null distribution. Examples of univariate samples include the Kolmogorov–Smirnov test and the Anderson-Darling test \citep{razali2011power}. 
Modern nonparametric methods, such as Maximum Mean Discrepancy (MMD) and Energy Distance \citep{gretton2012kernel,szekely2013energy}, use prespecified kernel embeddings or Euclidean distances to provide a distance metric between two distributions. These methods aim to address the question of how far two distributions are apart through a single scalar metric. As such, they cannot reveal exactly how and where they differ. One is often satisfied with a simulation based $p$-value for these metrics attained through resampling. 

In general, we question the relevance of the hypothesis testing perspective on two-sample comparisons in the context of assessing modern generative models---{\em all models are wrong, and all generative models are definitely quite wrong} in applications that involve highly complex data generative mechanisms. There is no question that the null hypothesis, with sufficient data (often available in synthetic experiments), {\em will} be rejected. The more valuable question is how the generative model fails to capture the truth. As such, we believe the function estimation perspective is more appropriate. 

Which function should we estimate as a way to summarize two-sample discrepancies? A natural candidate is the {\em density ratio}, i.e., the ratio of the two sampling densities. 
Indeed, density ratio estimation (DRE) has been studied, and well-known approaches include kernel-based methods such as KLIEP \citep{sugiyama2007direct} and uLSIF \citep{kanamori2009least}. More recently, \cite{awaya2025two} introduced an approach that uses additive tree ensembles along with a new loss function to learn the DR and achieves generalized Bayesian uncertainty quantification. 
A fundamental assumption of the DRE approach is that the two distributions must have the same support; otherwise, the ratio is undefined in parts of the space. Such an assumption is easily violated in very high-dimensional, complex data generative settings. This issue can incur severe extrapolation errors, a form of overfitting, when highly flexible function approximators, such as neural-networks, are adopted for DRE.

We circumvent this difficulty in direct DRE by considering the following modification to the density ratio:
\begin{align}
    r(x) = \frac{p(x)}{\frac{1}{2}\br{p(x)+q(x)}} \label{eq:RDR}
\end{align}
as our choice of the function-valued metric, which \cite{yamada2013relative} first considered and referred to as the {\em relative density ratio} (RDR). 
In our setting, the denominator corresponds specifically to the 50–50 mixture of $p$ and $q$. While \cite{yamada2013relative} is largely grounded in kernel-based methods for RDR estimation, our estimation approach is tailored to neural-network-based estimators using $\phi$-divergence objectives, following the general strategy introduced in \cite{Nguyen2010-xb}. We will show in Section~\ref{sec:RDR} that the RDR has a bounded range, good interpretability, and numerical stability for neural network estimators.

Alongside the proposed relative density ratio framework, we contribute new theoretical insights and empirical results on both image generation and scientific data domains, including the relative abundance microbiome data from the American Gut Project \citep{AmericanGut}.  We first show that the ratio estimator can be viewed as an M-estimator, and we establish consistency and convergence rate results for general estimator classes that satisfy certain entropy bounds. Specifically, for the sparse neural network estimator, building on prior work by \citet{suzuki2021deep}, we show that if the true ratio function lies in an anisotropic Besov space, a function space that allows differing degrees of smoothness across coordinates, then the convergence rate of the neural network estimator is nearly optimal up to a logarithmic factor in $n$, and depends only on the smoothness of the true ratio function, rather than directly on the sample dimension.

In our numerical experiments, we compare a range of modern generative models, including variational autoencoders (VAEs), generative adversarial networks (deep convolutional GANs, \citealp{DCGAN}), diffusion models (DDIM, \citealp{song2021ddim}), and continuous flow matching models (ICFM, \citealp{lipman2022flow}), on image datasets such as MNIST and CelebA-64, as well as on the American Gut data. We found that even for the well-performing generators with high fidelity and good coverage, in a simple task of generating digits from MNIST data, the generator tends to have its own bias and concentrate on certain local regions in the high-dimensional support. A similar trend also appears in the higher-dimensional CelebA-64 data and the more complex, highly sparse microbiome data. Furthermore, in the experiments of CelebA-64 and the American Gut Project, we demonstrate how to use this ratio function to detect the features or attributes associated with such distributional discrepancies, making it more interpretable.

The rest of this paper is organized as follows. Section~\ref{sec:RDR} introduces the properties of the relative density ratio and demonstrates its connection to existing evaluation criteria in a simple 1D setting. Section~\ref{sec:estimation_and_theory} provides the estimation methods and theoretical results. Section~\ref{sec:numerical_example} provides the numerical experiments on MNIST, CelebA-64, and the American Gut Project data. We conclude in Section~\ref{sec:discussion} with a discussion of potential future directions.

\section{Distributional Evaluation via Relative Density Ratio}\label{sec:RDR}
For a compact support $\cX$, let $X_1,\dots,X_n\sim P$ and $Y_1,\dots,Y_m\sim Q$. Let $a \lesssim b$ denote that $a \leq b$ up to a constant. Let $\leb$ be the base measure. Let $\bR$ be the real line, and $\bR^+$ be $[0,\infty)$. For the generative model evaluation task, we assume $P$ is the true distribution for the real sample, and $Q$ is the distribution for the generative model.

When designing an evaluation method, we aim for the following desiderata. Among them, (A) is a common focus of existing methods, but we argue is inadequate by itself, whereas (B) and (C) are additional objectives that are of high practical values which we also aim to achieve.
\begin{itemize}
    \item[(A)] A global numeric score for comparing and ranking different generators;
    \item[(B)] An observation-level metric that provides insights into the fit of a generator; 
    \item[(C)] Feature learning for summarizing distributional discrepancies, i.e., identifying features that contribute to the distributional difference.
\end{itemize}

The RDR, as defined in \eqref{eq:RDR}, along with several quantities that can be derived from it, satisfies all of the above desiderata, as later illustrated in 1D numerical examples in Section~\ref{subsec:connection}. Proposition~\ref{prop:RDR} shows two related properties that make RDR interpretable.

\begin{proposition}\label{prop:RDR}
The RDR $r(x)$ has a bounded image $[0,2]$. The squared Hellinger distance between the numerator density $p(x)$ and the denominator density $\frac{1}{2}\br{p(x) + q(x)}$ is bounded above by $1-\frac{1}{\sqrt{2}}(\approx 0.293)$.
\end{proposition}

For desideratum (A), we adopt the squared Hellinger distance $H^2(p,\frac{p+q}{2})$ as the global score: if $H^2(p,\frac{p+q}{2})=0$, the two distributions are identical; whereas if $H^2(p,\frac{p+q}{2})=1-\frac{1}{\sqrt{2}}$, there is almost no overlap between them. The approximated upper bound $0.293$ will repeatedly show up in our numerical examples in Section~\ref{subsec:connection} and Section~\ref{sec:numerical_example}. The boundedness will also lead to numerical robustness in estimation from the data. For desideratum (B), recall that the definition of RDR is $\frac{p(x)}{0.5p(x)+0.5q(x)} = \frac{2p/q}{p/q+1}$. For an observation $X_i=x$ in the support $\supp(P)\cup\supp(Q)$, the probability masses of the two distributions are equal when $r(x)=1$; if $r(x)$ is close to $0$, $X_i$ is more likely drawn from $Q$ than from $P$; and vice versa when $r(x)$ is close to $2$. For desideratum (C), samples evaluated by $r(x)$ form a one-dimensional distribution $\{r(X_i)\}_{i=1}^n$. When covariates or features are available for each observation, downstream analyses, such as regression or simple correlation analysis, can help identify covariates that drive the distributional differences. If we have two high-dimensional distributions $X\sim P$ and $Y\sim Q$, the following Theorem~\ref{thm:dpi} shows that the 1D distributions after the RDR mapping, $r(X)$ and $r(Y)$, preserve the $\phi$-divergence of the original high-dimensional distributions $P$ and $Q$. If $X\sim P$, given a measurable map $F$, we use the pushforward notation $F_{\#}P$ to denote the distribution of $F(X)$.

\begin{theorem}\label{thm:dpi}
Assume $P$ and $Q$ have common support $\cX$. The density ratio transformation $g:\cX\to\R^{+}$, $g(x) = p(x)/q(x)$, is a maximizer of $D_\phi(F_{\#}P \| F_{\#}Q)$ among all measurable maps $F:\cX\to \cY$. Consequently, for any one-to-one mapping $T$, $T_g:=T\circ g$ is also a maximizer of $D_\phi(F_{\#}P \| F_{\#}Q)$ among all measurable $F$ on $\cX$. Furthermore, $D_\phi({T_g}_{\#}P \| {T_g}_{\#}Q) = D_\phi(P \| Q)$.
\end{theorem}

Note that there is a one-to-one mapping between RDR and the density ratio: $r = T(g)$, where $T(g) = \frac{2g}{g+1}$, which is a one-to-one mapping between $r(x)$ and $g(x)$. Because of the one-to-one mapping, Theorem~\ref{thm:dpi} 
indicates that we can use the 1D empirical distributions formed by $\br{r(X_i)}_{i=1}^n,\br{r(Y_j)}_{j=1}^m$ to summarize the distributional difference between $P$ and $Q$ in high-dimensions while preserving the $\phi$-divergence $D_\phi(P\| Q)$. 
We can use the 1D histograms of $\br{r(X_i)}_{i=1}^n$, $\br{r(Y_j)}_{j=1}^m$ as the distributional visualization of the original two samples; histograms such as Row 3 in Figure~\ref{fig:oneD_beta} demonstrate the coverage/fidelity of the generator of high-dimensional samples.

We first introduce the estimation method of the ratio function $r(x)$ in Section~\ref{subsec:estimation}, and then illustrate the usage of RDR in Section~\ref{subsec:connection} with simple 1D examples in connection with existing two-sample comparison methods.

\subsection{Estimation via the variational form of $\phi$-divergence}\label{subsec:estimation}

Our estimation of $r(x)$ is closely related to the density ratio $g(x) = p(x)/q(x)$. If $P\ll Q$ and $\phi:\R^{+}\to \R$ is a convex function with $\phi(1)$=0, let $D_\phi$ be the general $\phi$-divergence, defined as $D_\phi(P\| Q):=\int_\cX \phi\sbr{\frac{\dd P}{\dd Q}(x)}\dd Q(x)$. Table~\ref{tb:phi_div} lists a few common options for $\phi$. The variational form  (dual form) of $\phi$-divergence provides a way to learn the density ratio $g(x)$. This theoretical property of $\phi$-divergence is proved in \cite{Nguyen2010-xb}. In convex optimization \citep{rockafellar1997convex}, the \textbf{subdifferential} of a convex function $\phi$ at point $t\in\R$ is defined as the set 
\[\partial \phi(t):=\br{z\in\R:\phi(s)\geq \phi(t) + z(s-t)}\]
and if $\phi$ is differentiable at $t$, $\partial \phi(t)=\br{\phi'(t)}$. The \textbf{conjugate dual function}, denoted as $\phi^c$, associated with $\phi$, is defined as 
\[\phi^c:=\sup_{u\in\R}\br{uv-\phi(u)}.\]

We reiterate Lemma 1 in \cite{Nguyen2010-xb} since it is the foundation of our estimation procedure.

\begin{lemma}[\citet{Nguyen2010-xb}]\label{lem:long_lem1}
    For any $\phi$-divergence, the dual form can be written as an optimization problem with respect to a function class $\cF$ mapping from $\cX$ to $\R$,
    \begin{equation}\label{eq:variational_form}
        D_\phi(P\|Q) \geq \sup_{f\in\cF}\int [f \dd P - \phi^c(f)\dd Q]
    \end{equation}
    The equality holds when the optimal $f^*\in \partial \phi(p/q)$ for any $x\in\cX$. 
\end{lemma}

This result indicates that the optimizer of the variational form of $\phi$ divergences is a certain function of the density ratio $g$, leading to many ways of estimating the density ratio. Table \ref{tb:phi_div} shows the optimizers for KL-divergence, chi-squared divergence, and the squared Hellinger divergence. In fact, this idea has been explored by \cite{nowozin2016f} to use the density ratio as the discriminator in a GAN (Generative Adversarial Network). 

Due to the numerical stability and interpretability of the bounded range, we choose the squared Hellinger distance because it provides a bounded outcome, as discussed in Proposition~\ref{prop:RDR}. 
\begin{table}[ht!]
    \centering
    \resizebox{\textwidth}{!}{
    \begin{tabular}{@{}lll@{}}
    \toprule
    \textbf{$\phi$--divergence} & \textbf{Variational form} & \textbf{Optimizer $f^*(x)$} \\
    \midrule
    KL: $\phi(u)=u\log u$
    &
    $\displaystyle D_{\mathrm{KL}}(P\|Q)
    =\sup_{f}\Big\{
    \bE_{P}[f(X)] \;-\;\bE_{Q}\big[e^{f(X)-1}\big]
    \Big\}$
    &
    $\displaystyle f^*(x)=1+\log\frac{p(x)}{q(x)}$ \\[2ex]
    
    Chi‐squared: $\phi(u)=(u-1)^2$
    &
    $\displaystyle \chi^2(P\|Q)
    =\sup_{f}\Big\{
    \bE_{P}[f(X)] \;-\;\bE_{Q}\big[\tfrac{f(X)^2}{4}+f(X)\big]
    \Big\}$
    &
    $\displaystyle f^*(x)=2\Bigl(\frac{p(x)}{q(x)}-1\Bigr)$ \\[2ex]
    
    Squared Hellinger: $\phi(u)=\frac{1}{2}(\sqrt{u}-1)^2$
    &
    $\displaystyle H^2(P\|Q)
    =\sup_{f<\frac{1}{2}}\Big\{\bE_{P}[f(X)] \;-\;\bE_{Q}\!\Big[\frac{f}{1-2f}\Big]
    \Big\}$
    &
    $\displaystyle f^*(x)=\frac{1}{2}-\frac{1}{2}\sqrt{\frac{q(x)}{p(x)}}$ \\
    \bottomrule
    \end{tabular}
    }
\caption{Variational representations of selected $\phi$-divergences and their optimizers. }
\label{tb:phi_div}
\end{table}

\subsection{A neural network model to estimate the (relative) density ratio}

When $X_1,\dots,X_n\overset{\iid}{\sim}P$, and $Y_1,\dots, Y_m\overset{\iid}{\sim}Q$, for any measurable set $A\subset \cX$, denote the empirical distributions $\bP_n(A):=\frac{1}{n}\sum_{i=1}^n \bfone\br{X_i\in A}$, and $\bQ_m(A):=\frac{1}{m}\sum_{j=1}^m\bfone\br{Y_j\in A}$. Using a change-of-variable transformation $g = (1-2f)^{-2}$ from the variational form of the squared Hellinger distance in Table~\ref{tb:phi_div}, let $\hat g_n$ be the estimated maximizer over a function class of candidates $\cG$, based on the observed samples, 
\begin{equation}\label{eq:objective}
    \hat g_n = \arg\max_{g\in \cG} \br{1-
    \frac{1}{2}\int g^{-\frac{1}{2}} \dd\bP_n -\frac{1}{2}\int g^{\frac{1}{2}}\dd\bQ_m}
\end{equation}
For the estimation of RDR, $r(x)$ defined in \eqref{eq:RDR}, we need to replace $\bQ_m$ with the mixture $\widetilde Q_m:=0.5\bP_n+0.5\bQ_m$.
\begin{equation}\label{eq:RDR_objective}
    \hat r_n = \arg\max_{g\in \cG} \br{1-
    \frac{1}{2}\int g^{-\frac{1}{2}} \dd\bP_n -\frac{1}{2}\int g^{\frac{1}{2}}\dd\widetilde\bQ_m}.
\end{equation}
The estimation method depends on the class of candidate functions in $\cG$ and the smoothness of the true function $g_0$. \citet{awaya2025two} considers additive tree ensembles for $\log g$ as their $\cG$. Note that $g^*=p/q$ is the unique global minimizer of $\frac{1}{2}\int g^{-\frac{1}{2}} \dd\bP +\frac{1}{2}\int g^{\frac{1}{2}}\dd\bQ$. 
In this work, we focus on using neural networks to estimate the ratio functions under this loss through gradient-based optimization. 
In practice, (stochastic) gradient descent algorithms to find the optimal $\hat g_n,\hat r_n$.

The feedforward neural network can be viewed as a function learner \citep{Sazli2006BriefReview}, parameterized with weight matrices $\br{W^{(l)}}_{l=1}^L$ and bias vectors $\br{b^{(l)}}_{l=1}^L$. Let $U_l$ be the total number of nodes in layer $l$ when $1<l\leq L$, $W^{(l)}\in\R^{U_l\times U_{l-1}}$, and the vector-valued activation function $\sigma^{(l)}:\R^{U_l}\to\R^{U_l}$ is an element-wise Lipschitz function. Denote $\bftheta:=\br{W^{(l)},b^{(l)}}_{l=1}^L$ as the set of all parameters, a neural network estimator can be formulated as a sequence of compositions 
$$f(x;\bftheta):=\sigma^{(L)}(W^{(L)}(\cdot)+b^{(L)})\circ\dots \circ \sigma^{(2)}(W^{(2)}(\cdot)+b^{(2)})\circ \sigma^{(1)}(W^{(1)}x+b^{(1)}),$$ see more details in Section~\ref{subsubsec:ffnn}. Using a neural network estimator to learn the density ratio becomes an optimization problem
\begin{equation*}
    \widehat\bftheta_g = \arg\min_{\bftheta_g} \br{\frac{1}{2}\sum_{i=1}^n f^{1/2}(X_i;\bftheta_g) + \frac{1}{2} \sum_{j=1}^m f^{-1/2}(Y_j;\bftheta_g)}.
\end{equation*}
The empirical loss being minimized was first introduced in \cite{awaya2025two} and referred to as the ``balancing loss'' for reasons explained in that paper. The optimizer in \eqref{eq:objective} can be estimated by $\hat g_n(x) = f(x;\widehat\bftheta_g)$. Similarly, given the objective in \eqref{eq:RDR_objective} and the learned $\widehat\bftheta_r$, $\hat r_n =f(x;\widehat\bftheta_r)$.
The neural network architecture varies in different data applications: for image data, variations of convolutional neural networks (CNNs) are applied, and for lower-dimensional examples and microbiome data, multilayer perceptrons (MLPs) are sufficient. Later in Section~\ref{sec:estimation_and_theory}, we will first provide a general convergence theory for general estimators with constraints on $\cG$, followed by a convergence rate result using the sieved sparse neural networks space as the function class $\cG_n$, where the size of the neural network space is allowed to grow with the sample size $n$.

\subsection{Connection to existing methods}\label{subsec:connection}

\noindent\textbf{Connection to DRE.}

The main appeal of using $r(x)$ rather than the density ratio $g(x)$ is that, with a mixture of $p(x)$ and $q(x)$, the denominator measure is guaranteed to have full support over the numerator measure $p(x)$, i.e., $P \ll \frac{P+Q}{2}$, whereas the same is not guaranteed for the density ratio $g(x)$, especially in high-dimensional problems. Even though we could assume in theory fully overlapping support for $p$ and $q$; in practice, even for simple 1D Gaussians, the estimated density ratio $p/q$ can be far less than the theoretical value when $p$ has a nonzero measure but $q$ is near-zero. The theoretical properties related to $r(x)$ in Proposition~\ref{prop:RDR} yield substantial improvements in numerical estimation compared to $g(x)$.
Figure~\ref{fig:oneD_compare} shows the 1D normal comparison of $r(x)$ and $g(x)$ as the mean shift increases. We use a simple 4-layer MLP to learn both $r(x)$ and $g(x)$, and train the function estimator on the observed training data with 1000 independent draws from both groups, and evaluate the trained function $r(x)$ or $g(x)$ over an evenly-spaced 1D grid of 500 grid points on $[-6,6]$. 

We observe that for both $\hat g(x)$ and $\hat r(x)$, there exist estimation errors. There are two potential culprits. First, neural network estimators, such as a simple ReLU MLP, extrapolate linearly to out-of-sample areas \citep{xu2021how}, which leads to estimation error in the low data mass region (see Figure~\ref{fig:oneD_compare} Column 1, $\hat r(x)$ has large extrapolation error when $x>3$). Second, the estimation performance of $r(x)$ depends on the output activation function of the last layer. There are different choices of the output activation function that all have a bounded image $[0,2]$, and the estimation performance depends on how fast the output activation function saturates to the end points $0,2$. In the Supplementary Section~\ref{supp_sec:act_sensi}, we provide a discussion on the choice of output activation functions and a sensitivity analysis on their performances.

\begin{figure}[ht]
    \centering
    \includegraphics[width=\linewidth]{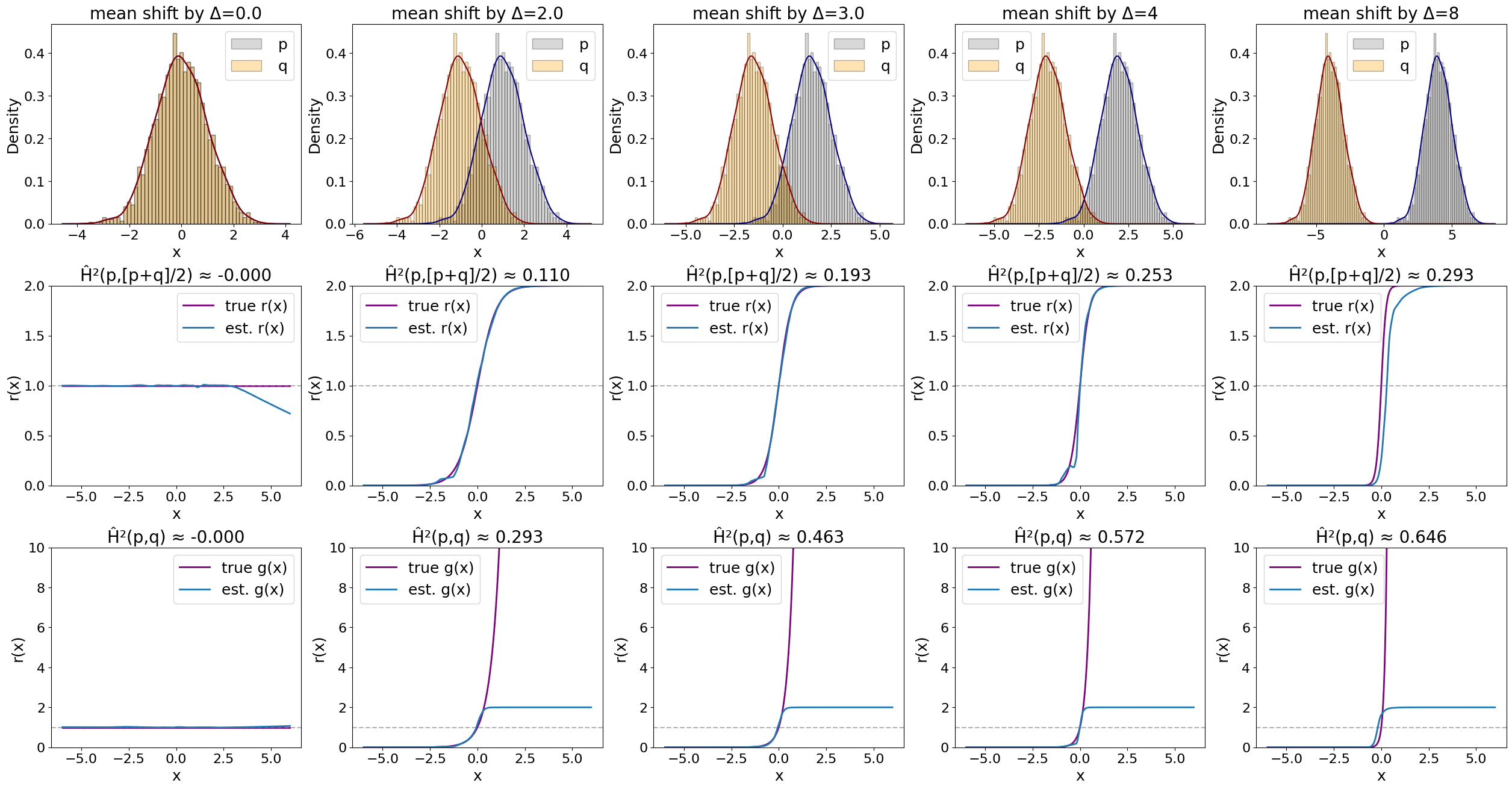}
    \caption{1D normal comparison between the relative density ratio $r(x)$ (Row 2) and the density ratio $g(x)$ (Row 3). Both function estimators are learned using the same ReLU MLP.}
    \label{fig:oneD_compare}
\end{figure}

 In Figure~\ref{fig:oneD_compare}, as the normal densities move further from each other, the theoretical density ratio $g(x)$ quickly explodes to infinity, whereas the relative density ratio $r(x)$ is theoretically bounded (see Proposition~\ref{prop:RDR}), and is more robust to extrapolation errors in areas with low data mass. This is the numerical advantage brought by $r(x)\in [0,2]$ in Proposition~\ref{prop:RDR}.

Recall desideratum (A), when $p$ and $q$ are different,  we need to measure how different they are. The objective value of \eqref{eq:RDR_objective} automatically provides the estimate for the squared Hellinger distance between $p$ and $(p+q)/2$. Hence, we use the squared Hellinger distance as the global numeric score in desideratum (A), as illustrated in Figure~\ref{fig:oneD_compare}. Both the estimated $H^2(p,q)$ and $H^2(p,\frac{p+q}{2})$ increase as the two densities increasingly shift away. The theoretical upper bound for $H^2(p,q)$ is 1, while the upper bound for $H^2(p,\frac{p+q}{2})$ is $1-1/\sqrt{2}$ (shown in Proposition~\ref{prop:RDR}). In the far-right panel in Figure~\ref{fig:oneD_compare}, when the normal means are shifted by 8, the empirical distributions are almost completely non-overlapping, which should yield the theoretical upper bound for both $H^2(p,q)$ and $H^2(p,\frac{p+q}{2})$. However, because of the massive extrapolation error induced by $\hat g(x)$, the estimated $\hat H^2(p,q)$ is still far from 1, whereas the estimated $\hat H^2\sbr{p,\frac{p+q}{2}}$ is almost exactly at the theoretical upper bound.

\vspace{0.5em}

\noindent\textbf{Connection to binary classification and the ``density ratio trick''.} 

As one of the most classical two-sample comparison methods, binary classifiers center around estimating $P(Z=1|X)$, where $Z\in\br{0,1}$ is the group label, and $X$ is the sample of interest. If we assume $X\sim P$ when $Z=1$ and $X\sim Q$ when $Z=0$, 
\begin{align*}
    P(Z=1|X)&=\frac{P(X|Z=1)P(Z=1)}{P(X|Z=1)P(Z=1) + P(X|Z=0)P(Z=0)},
\end{align*}
then 
the density ratio can be recovered from the estimated $P(Z=1|X)$ through the ``density ratio trick'' because $g(X) = P(Z=1|X)/P(Z=0|X)\cdot P(Z=0)/P(Z=1)$ and $r=2g/(g+1)$.
However, training a binary classifier under misclassification risk can be sensitive to the marginal sample probabilities $P(Z=i),i\in\br{0,1}$,
and \cite{awaya2025two} demonstrate numerically that the density ratios estimated with binary classifiers can be severely biased when the two sample sizes are  not equal and/or when the support of the two densities have regions of poor overlap. 
Estimation of $r(x)$ through the variational form of $\phi$-divergence is generally robust to unbalanced sample sizes and poor support overlap. This is highly relevant because both are common scenarios in applications involving generative models. 

\vspace{0.5em}

\noindent\textbf{Connection to precision-recall evaluation.}

A recent popular evaluation criterion for generative models is the precision-recall curve \citep{pmlr-v97-simon19a,sajjadi2018assessing}. As illustrated in Figure~\ref{fig:oneD_beta}, in the context of Figure 1 in \citet{pmlr-v97-simon19a}, precision measures how much of $Q$ can be generated by a part of $P$, and recall measures how much of $P$ can be generated by a part of $Q$. The left panel in Figure~\ref{fig:oneD_beta} (partial precision, full recall) indicates that the generator has full coverage (all $P$-mass within $\supp(Q)$), but suffers from potential model hallucination ($\supp(Q)\subsetneq\supp(P)$). The middle panel is the reversed case, where the generator has high fidelity but only partial coverage ($\supp(Q)\subsetneq\supp(P)$), missing critical regions in the real data, and the right panel is a low-fidelity and low-coverage case ($\supp(P)\nsubseteq\supp(Q)$, $\supp(Q)\nsubseteq\supp(P)$). The second row in Figure~\ref{fig:oneD_beta} shows the theoretical and estimated $r(x)$, which gives a clear indication of the precision-recall/coverage-fidelity situation in each case; moreover, $r(x)$ pinpoints exactly which areas in the support suffer from low-fidelity (when $r(x)$ is near 0) and which areas suffer from low-coverage (when $r(x)$ is near 2). The bottom row of Figure~\ref{fig:oneD_beta} shows the histogram of $r(X_i)$ evaluated on the training data, where $X_i$ is drawn from $P$ or $Q$. The histogram also shows a clear range of $r(x)$: for the model hallucination (low fidelity) case (left), the range of $r(x)$ is strictly bounded below 2 with a large mass near 0. In the reversed case (middle), $r(x)$ is bounded above 0 with large mass near 2. More interestingly, in the case of partial precision and partial recall (right), we see a spike of mass of $r(X_i),X_i\sim P$ concentrated around 2, another spike of $r(X_i),X_i\sim Q$ concentrated around 0, and a smaller group of samples with overlapping $r(X_i)$ from both groups in the middle. This histogram plot will be useful in high-dimensional evaluation to pinpoint regions of low fidelity or coverage and regions of good representation when it is expensive to create evenly spaced grids. 

\begin{figure}[ht]
    \centering
    \includegraphics[width=0.6\linewidth]{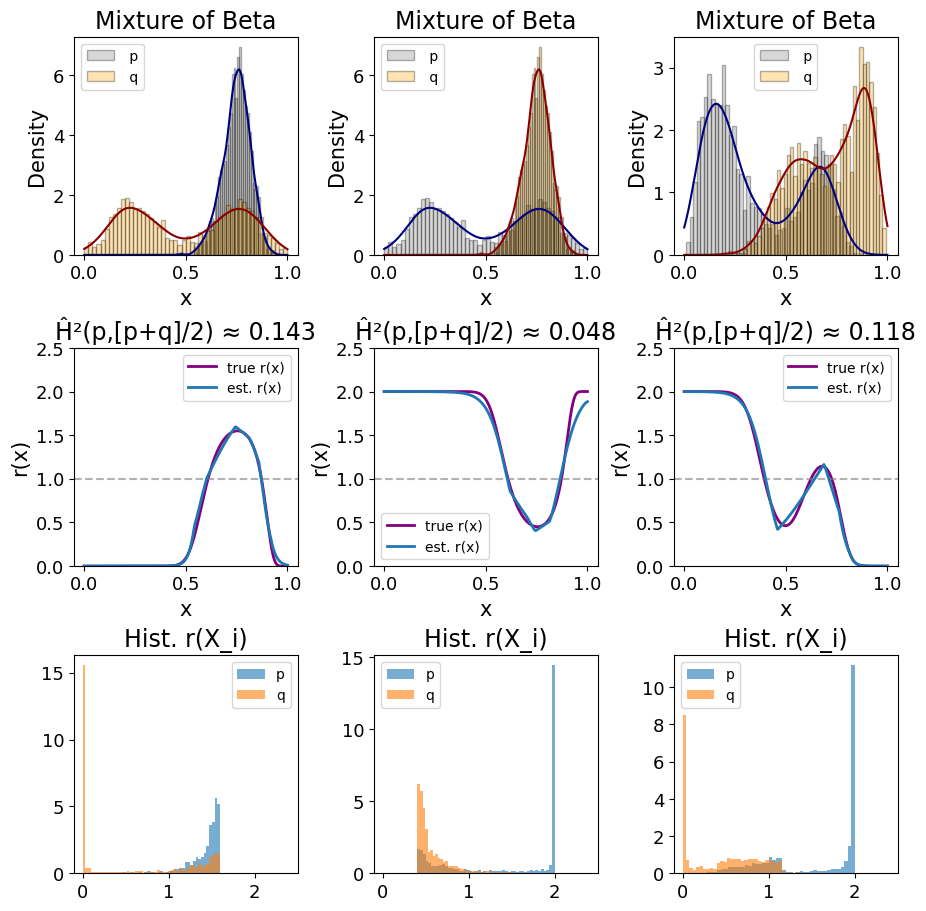}
    \caption{A 1D 3-mixture of Beta densities demonstrating the precision-recall comparison, simulation of Figure 1 in \cite{pmlr-v97-simon19a}. Left (model hallucination, low fidelity): partial precision, full recall. Middle (low coverage): partial recall, full precision. Right: partial recall and precision, a mixed case.}
    \label{fig:oneD_beta}
\end{figure}

\section{Theoretical Analysis}\label{sec:estimation_and_theory}

Based on the M-estimator theory \citep{van1996weak}, with the objective function being the variational form of the squared Hellinger distance \eqref{eq:objective}, we first provide the general consistency and convergence rate results for estimators that satisfy a certain entropy bound of the function space in Section~\ref{subsec:convergence_general}. Then we explore the sparse neural network function space with 1-Lipschitz activation functions such as ReLU and provide a convergence rate result when the true ratio function is assumed to be in the anisotropic Besov space. In this section, we use $p_0,q_0$ to specifically denote the true density functions for distributions $P,Q$. Denote $g_0$ as the true density ratio $g_0 = p_0/q_0$, and $r_0$ as the true relative density ratio $r_0=\frac{2p_0}{p_0+q_0}$.

\subsection{General convergence theory}\label{subsec:convergence_general}

The following theorem establishes the almost sure consistency of the density ratio function as an M-estimator for the variational form of the squared Hellinger distance. The main proof techniques follow directly from \cite{Nguyen2010-xb}, with the adaptation to find a metric that can lower bound the difference between functions induced by the Hellinger distance.

Let $\cG$ be the general function class representing the search space for the optimization problem in the variational form. As discussed in Section~\ref{sec:RDR}, the variational form of the squared Hellinger divergence in Table \ref{tb:phi_div} can be written as 

\begin{equation}\label{eq:hellinger_objective}
\sup_{g>0} \br{1-\frac{1}{2}\int g^{-\frac{1}{2}} \dd\bP - \frac{1}{2}\int g^{\frac{1}{2}}\dd\bQ}.    
\end{equation}

The optimizer for this objective function is at $g_0=p_0/q_0$. Let $\cG^{1/2}:=\br{g^{1/2}:g\in\cG}$ and $\cG^{-1/2}:=\br{g^{-1/2}:g\in\cG}$. Let $\cG^{1/2}-g_0^{1/2}$ be the shifted function space $\br{g^{1/2}-g_0^{1/2}:g\in\cG}$, similarly $\cG^{-1/2}-g_0^{-1/2}$ the shifted space of $\cG^{-1/2}$. To describe the size of the search space $\cG$, we denote $N(\delta,\cG,L^r(\bP))$ as the smallest covering number of the space $\cG$ with metric $L^r(\bP)$ and radius $\delta$. We introduce the following assumptions for the general convergence results. 

\begin{assumption}\label{asm:sample_size_ratio}
    Assume that $\frac{m}{n}\to \rho$ as $m,n\to\infty$, where $\rho>0$ is a finite constant.
\end{assumption}

\begin{assumption} \label{asm:gen_consist_true} 
    The function class $\cG$ contains some $g$ such that $g=g_0$ almost surely.
\end{assumption}

\begin{assumption} \label{asm:g_L1_bound}
    Define envelope functions $G_0(x):=\sup_{g\in \cG}|g^{1/2}(x)|$, $G_1(x):=\sup_{g\in\cG}|g^{-1/2}(x)|$. Assume $\int G_0\dd\bQ< \infty$, $\int G_1\dd\bP< \infty$. Denote $M_0:=\int G_0\dd\bQ< \infty$.
\end{assumption}

\begin{assumption}\label{asm:general_entropy}
    Assume that for any $\delta>0$, 
        \begin{align*}
            &\frac{1}{m}\log N(\delta, \cG^{1/2}-g_0^{1/2},L^1(\bQ_m)) \overset{\bQ}{\to} 0\\
            &\frac{1}{n}\log N(\delta, \cG^{-1/2}-g_0^{-1/2},L^1(\bP_n)) \overset{\bP}{\to} 0
        \end{align*}
\end{assumption}

Assumption~\ref{asm:gen_consist_true} is used for the consistency proof. Assumption~\ref{asm:g_L1_bound} is a standard assumption in empirical process theory about the envelop functions, and Assumption~\ref{asm:general_entropy} is the entropy condition for the general estimator space $\cG$. Similar assumptions as Assumptions~\ref{asm:gen_consist_true} to \ref{asm:general_entropy} are made in \citet{Nguyen2010-xb}.

Recall the squared Hellinger distance is denoted as $H^2(\bP\|\bQ):=\frac{1}{2}\int (\sqrt{\dd\bP}-\sqrt{\dd\bQ})^2$. The corresponding sample-level objective value to \eqref{eq:hellinger_objective} is denoted as 
\begin{equation}
    \hat H^2 = 1 - \frac{1}{2}\int \hat g_n^{-1/2} \dd\bP_n -  \frac{1}{2}\int \hat g_n^{1/2} \dd\bQ_m
\end{equation}

\begin{theorem}\label{thm:consistency}
    Under Assumption \ref{asm:sample_size_ratio} to \ref{asm:general_entropy}, $|\hat H^2-H^2(\bP\|\bQ)|\overset{\as}{\to}0$, $\|\hat g_n^{1/2}-g_0^{1/2}\|_{L^1(\bQ)}\overset{\as}{\to}0$ and $\|\hat g_n^{-1/2}-g_0^{-1/2}\|_{L^1(\bP)}\overset{\as}{\to}0$.
\end{theorem}
This proof of Theorem \ref{thm:consistency} follows the proof of Theorem 1 in \cite{Nguyen2010-xb}, whose theory is based on the Kullback-Leibler divergence, with adaptation to the squared Hellinger distance.

Note that Assumption~\ref{asm:general_entropy} is a commonly used assumption when $\cG^{1/2}-g_0^{1/2}$ is replaced by some general function class $\cF$. However, this assumption characterizes the size of the function class in a probabilistic statement through empirical measures $\bP_n,\bQ_m$. In the following Lemma~\ref{lem:entropy_bound}, we show that a bounded bracketing number with the population metric $L^1(\bP)$ implies that the entropy number with the empirical metric $L^1(\bP_n)$ is of the order $o_P(n)$. 

A bracketing number of a function class $\cF$ is defined as follows: for a metric $\|\cdot\|$, given any small $\delta>0$, a $\delta$-bracket is defined as a pair of functions such that $l\leq f\leq u$ and $\|l-u\|\leq \delta$. The bracketing number $N_{[]}(\delta,\cF,\|\cdot\|)$ is defined as the minimum number of $\delta$-brackets needed to cover $\cF$.

\begin{lemma}\label{lem:entropy_bound}
    For a general function class $\cF$ and for any $\epsilon>0$, $N_{[]}(\epsilon, \cF, L^1(\bP))<\infty$ implies $\frac{1}{n}\log N(\epsilon, \cF, L^1(\bP_n))\overset{\bP}{\to}0$.
\end{lemma}

Lemma \ref{lem:entropy_bound} implies that Assumption~\ref{asm:general_entropy} can be replaced by $N_{[]}(\epsilon, \cG^{1/2}-g_0^{1/2}, L^1(\bP))<\infty$ and $N_{[]}(\epsilon, \cG^{-1/2}-g_0^{-1/2}, L^1(\bP))<\infty$.

Based on the general consistency result in Theorem~\ref{thm:consistency}, we can show the convergence rate of $g$ with additional assumptions.

Define the squared Hellinger metric for functions with respect to measure $\bQ$ as 
\begin{equation}
    h^2_\bQ(g,g_0):=\int (\sqrt{g} - \sqrt{g_0})^2\dd\bQ.
\end{equation}
In the proof of the convergence rate, we need to characterize the following two function spaces:
\begin{equation}\label{eq:def_F}
    \cF_{\delta}^P:=\br{ g^{-1/2} - g_0^{-1/2} : g\in \cG, h_\bQ (g,g_0)<\delta }, \quad 
    \cF_{\delta}^Q:=\br{ g^{1/2} - g_0^{1/2} : g\in \cG, h_\bQ (g,g_0)<\delta  }. 
\end{equation}
We follow the notations in \cite{van1996weak} and denote the bracketing entropy integral as follows:
\[\tilde J_{[]}(\delta,\cF,\|\cdot\|):=\int_0^\delta \sqrt{1+\log N_{[]}(\epsilon,\cF,\|\cdot\|)} \dd \epsilon.\]

The size of $\cF_{\delta}^P$ and $\cF_{\delta}^Q$ will determine the convergence rate. We make the following assumptions to characterize the function spaces.

\begin{assumption} \label{asm:g_bound}
    Assume that $\sup_{g\in\cG}|g| < M_u <\infty$, $\sup_{g\in\cG}|g^{-1}| < M_l<\infty$ for some constants $M_u,M_l$.
\end{assumption}
Assumption~\ref{asm:g_bound} is stronger than Assumption~\ref{asm:g_L1_bound}. A similar assumption is also required for the convergence rate in \citet{Nguyen2010-xb}. 
\begin{assumption}\label{asm:general_entropy_integral} 
 Assume that for any $\epsilon>0$, $\log N_{[]}(\epsilon,\cF_\eta^P,L^2(\bP)) = O(\epsilon^{-\gamma_p})$, and $\log N_{[]}(\epsilon,\cF_\eta^Q,L^2(\bQ)) = O(\epsilon^{-\gamma_q})$ for some positive constants $\gamma_p<2,\gamma_q<2$, and some fixed small positive constant $\eta$.
\end{assumption}

\begin{theorem}\label{thm:convergence_rate}
    Under Assumptions \ref{asm:sample_size_ratio}, \ref{asm:gen_consist_true}, and \ref{asm:general_entropy} to \ref{asm:general_entropy_integral}, $h_\bQ(\hat g_n,g_0) = O_p(n^{-\frac{1}{2+\gamma}})$ where $\gamma = \min\br{\gamma_p.\gamma_q}$.
\end{theorem}

The above consistency and convergence rate results of the density ratio $g(x)$ can be directly applied to the relative density ratio $r(x)$, since throughout the proof of Theorem~\ref{thm:consistency} and Theorem~\ref{thm:convergence_rate}, we do not assume any independence between $\bP_n$ and $\bQ_m$. For the result of $r(x)$, we can view it as a type of density ratio by simply changing $\bQ$ to any general mixture $\widetilde\bQ:=\alpha\bP+(1-\alpha)\bQ, \alpha\in (0,1)$, similarly for $\bQ_m$. In our case, $\alpha=0.5$.

\begin{corollary}\label{col:general_rdr}
Under Assumptions \ref{asm:sample_size_ratio}, \ref{asm:gen_consist_true}, and \ref{asm:general_entropy} to \ref{asm:general_entropy_integral}, when replacing $\bQ$ by $\widetilde\bQ$, we have $h_{\widetilde\bQ}(\hat r_n,r_0) = O_p(n^{-\frac{1}{2+\gamma}})$ where $\gamma = \min\br{\gamma_p.\gamma_q}$.
\end{corollary}

\textbf{Remark.} Due to the boundedness of $r(x)$ as shown in Proposition~\ref{prop:RDR}, the upper boundedness assumptions in Assumptions~\ref{asm:g_L1_bound} and \ref{asm:g_bound} are no longer needed in Corollary~\ref{col:general_rdr}.

\subsection{Sieved Neural Network-based M-estimators}

The above section provides convergence rate results for general estimators, as long as the function space of the estimator satisfies the entropy conditions in Assumption~\ref{asm:general_entropy_integral}. In this section, we provide a convergence rate result in terms of sparse neural networks and assume that the true density $g_0$ is in a very general function class, specifically the anisotropic Besov space \citep{nikol2012approximation,triebel2006theory,suzuki2021deep}. Our result is built on \cite{suzuki2021deep}, which provides the approximation error of the sparse neural network for functions in the anisotropic Besov space. In this section, we assume that the support of the density ratio is on $\cX=[0,1]^d$.

\subsubsection{A Brief Review on the Anisotropic Besov Space}

We first introduce the anisotropic Besov space, following the same notations as in  \cite{suzuki2021deep}.

\begin{definition}[Anisotropic Besov space $(B^{\beta}_{s,t}(\cX))$]
    For any function $f\in L^s(\cX)$ where $s\in (0,\infty]$, the $r$-th modulus of smoothness of $f$ is defined as 
    \begin{align*}
        w_{r,s}(f,t) &= \sup_{h\in \R^d:|h_i|<t_i} \|\Delta_h^r(f)\|_{L^s(\cX)}, \numberthis \label{eq:modulus_of_smoothness}
    \end{align*}
    where $\Delta_h^r(f)(x) := \Delta_h^{r-1}(f)(x+h) - \Delta_h^{r-1}(f)(x)$, $\Delta_h^0(f)(x) := f(x).$
    For $s,t\in (0,\infty]$, the smoothness vector $\beta=(\beta_1,\dots,\beta_d)^\rT\in\R^d_{++}$, $r=\max_i \lfloor \beta_i \rfloor +1$, define the seminorm $|\cdot|_{B_{s,t}^\beta}$ as 
    \begin{align*}
        |f|_{B^{\beta}_{s,t}} := 
        \begin{cases}
        \left( \sum_{k=0}^{\infty} \big[ 2^k w_{r,s}\!\left(f, (2^{-k/\beta_1}, \ldots, 2^{-k/\beta_d})\right)\big]^t \right)^{1/t}, & (t < \infty), \\
        \sup_{k \geq 0} \, 2^k w_{r,s}\!\left(f, (2^{-k/\beta_1}, \ldots, 2^{-k/\beta_d})\right), & (t = \infty). \numberthis \label{eq:seminorm}
        \end{cases}
    \end{align*}
The anisotropic Besov space norm is defined as $\|f\|_{B^{\beta}_{s,t}}:=\|f\|_{L^s(\cX)} + |f|_{B^{\beta}_{s,t}}$. The anisotropic Besov space is defined as 
\begin{equation}\label{eq:asnisotropic_besov_space}
    B^{\beta}_{s,t}(\cX) = \br{f\in L^s(\cX):\|f\|_{B^{\beta}_{s,t}} < \infty} 
\end{equation}
\end{definition}

To understand this, the anisotropic Besov space $B^{\beta}_{s,t}(\cX)$ is controlled by $\beta$, the smoothness vector that allows for different smoothness $\beta_i$ in each direction; $s$, indicating that the smoothness is measured in $L^s(\cX)$; and $t$, used in the seminorm \eqref{eq:seminorm}. The seminorm \eqref{eq:seminorm} measures the localized anisotropic smoothness of $f$ in terms of the scale of oscillation. Unlike Sobolev spaces that measure smoothness through derivatives, the modulus of smoothness in \eqref{eq:modulus_of_smoothness} measures the finite difference given an incremental change $h$ and does not require differentiability. The seminorm \eqref{eq:seminorm} dyadically partitions the incremental change $h$ in different directions, $|h_i|<2^{-k/\beta_i}$, and $2^k$ is a normalization factor. Summing over $k$ aggregates the smoothness across different partition layers. \cite{triebel2011entropy} and \cite{suzuki2021deep} provide more detailed results on the relation between H\"older space and anisotropic Besov space. For example (Proposition 1 in \citet{suzuki2021deep}), when $\beta=(\beta_0,\dots,\beta_0)\in \bR^d$, and $\beta_0\notin \bN$, then  H\"older space $\cC^{\beta_0}=B^\beta_{\infty,\infty}$. 

The main advantage of assuming the true density ratio $g_0$ lies in $B^{\beta}_{s,t}(\cX)$ is that $B^{\beta}_{s,t}(\cX)$ allows for different smoothness in each direction. This coincides with our belief in two-sample comparisons in high dimensions, especially when comparing generative models: the two-sample difference may only lie in a few directions, especially when the generative models can produce realistic samples.

\subsubsection{Feedforward Neural Network Formulation}\label{subsubsec:ffnn}

Feedforward neural networks, including common architectures such as MLP and CNN, are usually formulated as a sequence of compositions in the following way. Let $U$ be the maximum number of nodes in each layer; $L$ is the total number of layers. Let $f^{(1)}(x) = \sigma\sbr{W^{(1)}x+b^{(1)}}$ be the first layer of the neural network for $x\in \R^d, W^{(1)}\in \R^{U\times d}, b^{(1)}\in \R^U$, where $\sigma$ is the element-wise activation function. We assume $\sigma$ is an element-wise 1-Lipschitz function and $\sigma(0)=0$; for example, the ReLU activation $\sigma(x) = \max\br{0,x}$ satisfies the element-wise 1-Lipschitz condition ($|f(x)-f(y)|\leq |x-y|$). Let $f^{(l)}(\cdot) = \sigma\sbr{W^{(l)}(\cdot)+b^{(l)}}$ be the $l$-layer, $W^{(l)}\in \R^{U\times U}, b^{(l)}\in \R^U$. For the last layer $f^{(L)}(\cdot) = \sigma\sbr{W^{(L)}(\cdot)+b^{(L)}}$, $W^{(L)}\in \R^{1\times U}, b^{(L)}\in \R$.

 We slightly abuse the notation and let $\|W\|_\infty:=\max_{i,j}|W_{i,j}|$, which is different from the matrix infinity norm. Similarly, $\|b\|_\infty:=\max_j|b_j|$. We define the following neural network function space characterized by the number of layers $L$, the maximum number of nodes in each layer $U$, the sparsity $S$, and the norm constraint $B$,
\begin{equation}\label{eq:def_nn}
    \Phi(L,U,S,B):=\br{f^{(L)}\circ \dots \circ f^{(1)}(x): \sum_{l=1}^L \sbr{\|W^{(l)}\|_0 + \|b^{(l)}\|_0} \leq S, \max_{l} \br{\|W^{(l)}\|_\infty, \|b^{(l)}\|_\infty} \leq B }
\end{equation}
This is a commonly used method of neural network parameterization in the literature \citep{suzuki2021deep,kaji2023adversarial}.

\subsubsection{Convergence Rate for Sieved Neural Network}

Denote $\cG_n$ as the sieve space of neural networks. The sieved M-estimator allows the search space $\cG_n$ to grow with sample size $n$, and is eventually flexible enough to approximate the true function (see details in Section 3.4 of \citet{van1996weak}). We define the sieve space according to the neural network parameter scales in Proposition 2 in \cite{suzuki2021deep}. Let $\cG_n  =\Phi(L_N,U_N,S_N,B_N)$, where the neural network parameters depend on $N$, and $N=N(n)$ depends on the sample size $n$. The dependence of $(L_N,U_N,S_N,B_N)$ on $N$ is specified in Proposition 2 in \cite{suzuki2021deep} and in our proof \eqref{eq:nn_params}. The rate $N(n)$ will be specified later in the proof.

Let $\cU\sbr{B^{\beta}_{s,t}(\cX)}$ be the unit ball of $B^{\beta}_{s,t}(\cX)$. Denote $\tilde\beta = \sbr{\sum_{j=1}^d 1/\beta_j}^{-1}$ as the harmonic mean of the smoothness vector $\beta$. 

\begin{theorem}\label{thm:sieve_nn_rate}
Under Assumption~\ref{asm:sample_size_ratio}, assume the functions in $\cG_n$ and $g_0$ satisfy the upper and lower bounds in Assumption~\ref{asm:g_bound} and $\|q_0^2/p_0\|_\infty <\infty$. If  $g_0\in \cU\sbr{B^{\beta}_{s,t}(\cX)}$ where $\tilde\beta>(1/s-1/2)_+$, then  the convergence rate for the sparse neural network density ratio estimator is 
    \begin{align*}
        h_\bQ(\hat g_n,g_0) = O_p(n^{-\frac{\tilde\beta}{2\tilde\beta+1}}\sbr{\log n}^{\frac{3\tilde\beta}{2\tilde\beta+1}}). \numberthis\label{eq:nn_rate}
    \end{align*}
\end{theorem}
\noindent\textbf{Remark.} The boundedness Assumption~\ref{asm:g_bound} can be enforced for neural network estimators by setting a bounded image of the activation function in the last layer. The unit ball assumption of $g_0$ is reasonable because we already assume the boundedness of $g_0$ in Assumption~\ref{asm:g_bound}.

Previously, in the proof of Theorem~\ref{thm:convergence_rate}, consistency is required; the true function $g_0$ must be in the same function space $\cG$ as the estimator; and the entropy condition is prespecified in Assumptions~\ref{asm:general_entropy} and \ref{asm:general_entropy_integral}. Here in Theorem~\ref{thm:sieve_nn_rate}, consistency is implied but not required, and we allow $g_0$ to be outside of $\cG_n$, although the distance between  $\cG_n$ and $\cU\sbr{B^{\beta}_{s,t}(\cX)}$ goes to 0 as $N\to\infty$ (see Proposition 2 of \cite{suzuki2021deep}); the entropy condition of the sparse neural network space $\cG_n$ is shown in Lemma~\ref{lem:nn_bracket_number} rather than being prespecified. Similar to Corollary~\ref{col:general_rdr}, since the proof of Theorem~\ref{thm:sieve_nn_rate} does not require the independence of $\bP_n$ and $\bQ_m$, we can directly apply this result to the general relative density ratio $r(x)$ by replacing $Q$ with $\widetilde Q=\alpha P + (1-\alpha) Q$, and similarly for $\widetilde Q_m$.

\begin{corollary}\label{col:nn_rdr_rate}
Under the same assumptions in Theorem~\ref{thm:sieve_nn_rate} but with $\bQ,\bQ_m$ replaced by $\widetilde\bQ$    and $\widetilde\bQ_m$, we have $h_{\widetilde\bQ}(\hat r_n,r_0) = O_p(n^{-\frac{\tilde\beta}{2\tilde\beta+1}}\sbr{\log n}^{\frac{3\tilde\beta}{2\tilde\beta+1}})$.
\end{corollary}

\noindent \textbf{Remark.} The convergence rate in \eqref{eq:nn_rate} depends only on the smoothness parameter $\tilde\beta$ and does not directly depend on the dimension $d$ of the sample space. This is because we allow the true function $g_0$ to have different smoothness in different directions, and the harmonic mean $\tilde\beta$ is only driven by a few least smooth directions, rather than assuming a spatially homogeneous smoothness bound in all directions (H\"older space), where the dimension $d$ will come into play, as in previous works \citep{nakada2020adaptive}.

\section{Numerical Experiments}\label{sec:numerical_example}

We use three numerical experiments to illustrate the rich information that can be extracted from the relative density ratio $r(x)$: the MNIST dataset \citep{mnist}, the CelebA-64 dataset \citep{CelebA64}, and the microbiome data from the American Gut Project \citep{AmericanGut}. Different from \cite{pmlr-v97-simon19a}, where the evaluation is conducted on the extracted inception features rather than the raw data, we perform all of our numerical experiments on the raw data to preserve the authentic distributional discrepancies. We provide simulation studies in Section~\ref{sec:sim} where the true relative density ratio function is known, and we compare the 4-layer MLP based on the Hellinger loss with other existing density ratio estimation methods such as uLSIF, KLIEP, and the \textit{density ratio trick} as discussed in Section~\ref{subsec:connection}. We use the bounded sigmoid function $\sigma_\alpha(x) = \frac{2}{1-\exp\br{-\alpha x}}$ with $\alpha=1$ as the default output activation function, and add a discussion and sensitivity analysis on the difference choice of output activation function in the Supplementary Section~\ref{supp_sec:act_sensi}.

\subsection{MNIST} \label{subsec:mnist}

The MNIST dataset \citep{mnist} consists of $60{,}000$ training and $10{,}000$ testing grayscale images, and each image, with dimensions $28\times 28$, is preprocessed so that all the voxels are mapped to $[-1,1]$. We test on two pretrained generators with saved checkpoints provided in the GitHub repository \citep{csinva2019ganvae}: VAE \citep{kingma2014vae} and DCGAN \citep{DCGAN}. To train the ratio function, we use batch updates and stochastic gradient descent to iterate through the training set; the architecture for the ratio function is the same as the discriminator in DCGAN, a CNN-based architecture. The ratio function is estimated using the training data, and we evaluate the ratio function on the $10{,}000$ real testing data and another $10{,}000$ generated samples from each generator. 

Figure~\ref{fig:MNIST_VAE} provides visualizations on $\{r(X_p),r(X_q)\}$ for VAE, and Figure~\ref{fig:MNIST_DCGAN} for DCGAN. Immediately from the scalar metric, $h^2(p,\frac{p+q}{2})$ is 0.277 for VAE, and 0.135 for DCGAN, we know that DCGAN is closer to the real sample than VAE. Note that this distance $h^2(p,\frac{p+q}{2})$ directly results from the negative training loss and does not require separate computation since the variational form is the dual problem of the $\phi$-divergence. The sample-level RDR allows us to dive much deeper than simply picking the best generative model.

\begin{figure}[ht!]
    \centering
    \includegraphics[width=\linewidth]{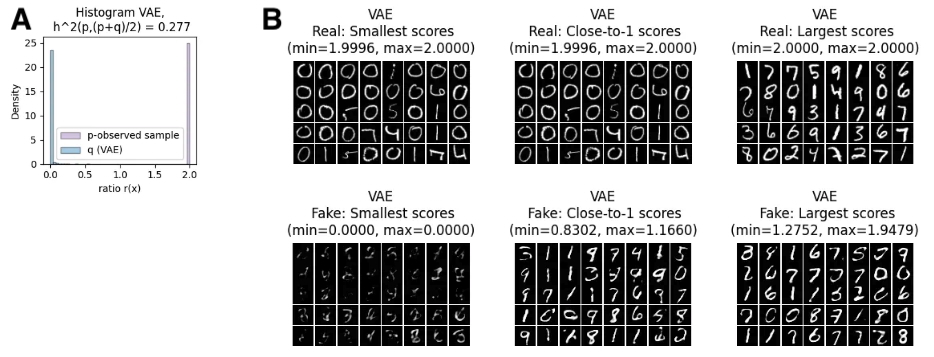}
    \caption{MNIST: VAE result. A: Histogram of $r(X_p),r(X_q)$ evaluated for the real test data and generated images. B: Top row: Real test images; Bottom row: Fake generated images; From left to right: images with $r(X)$ values from the smallest, closest to 1, and the largest; the range of $r(X)$ for the shown images is in each subplot title.
    }
    \label{fig:MNIST_VAE}
\end{figure}

Panel A in Figures~\ref{fig:MNIST_VAE} and \ref{fig:MNIST_DCGAN} shows the histogram of $\{r(X_p),r(X_q)\}$. The majority of the VAE $r(X_q)$ mass concentrates near 0 and has almost no overlap with $r(X_p)$, indicating low coverage and low fidelity. For DCGAN, 
the empirical support of $r(X_q)$ overlaps with $r(X_p)$ on $[0,2]$, 
indicating good coverage. But for DCGAN, there is still a noticeable portion of $r(X_q)$ centered around 0, even though this portion still overlaps with the mass of certain $r(X_p)$ samples, indicating a portion of the real sample is overly represented by the generator DCGAN. The sample-level RDR enables us to examine which regions of the support the DCGAN is biased toward. Because the real sample comes with digit labels, we can first check $r(X_p)$ stratified by digits, as shown in Panel B of Figures~\ref{fig:MNIST_DCGAN}. For DCGAN, although $r(X_p)$ ranges from 0 to 2 for all digits, digits 1 and 6 appear to have the most mass above 1, which means DCGAN tends to under-represent digits 1 and 6.

\begin{figure}[ht!]
    \centering
    \includegraphics[width=0.8\linewidth]{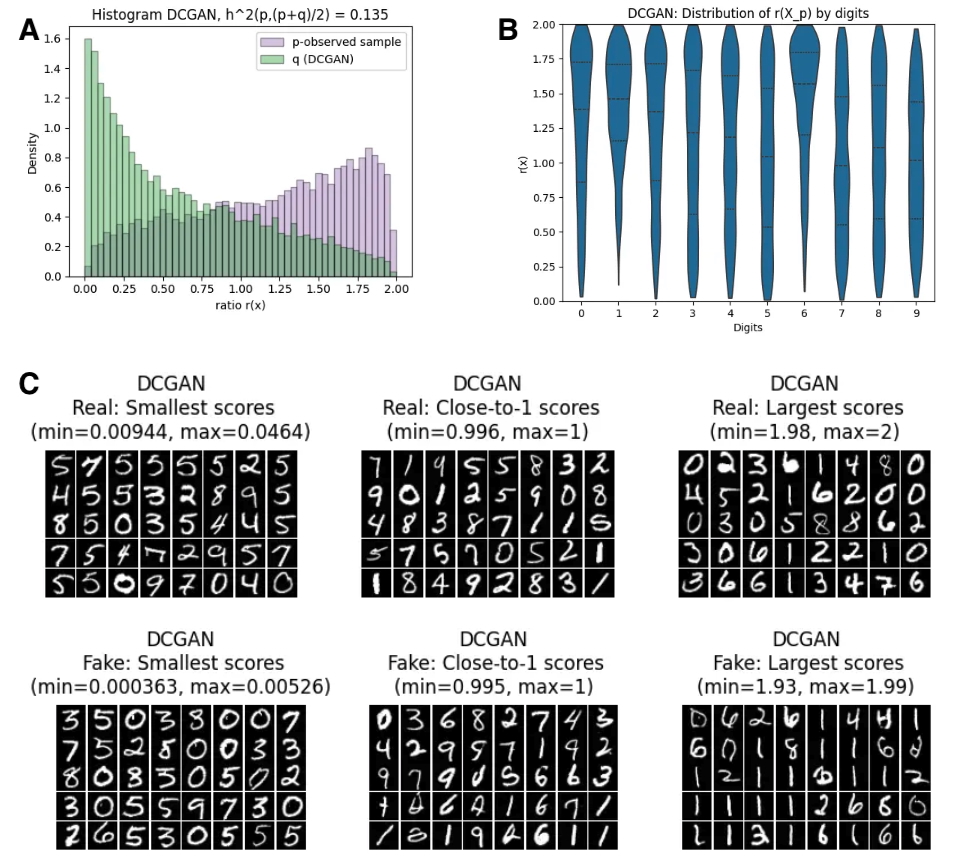}
    \caption{MNIST: DCGAN result. See explanations for the three groups of plots A, B, and C in the caption of Figure~\ref{fig:MNIST_VAE}. A: Histogram of $r(X_p),r(X_q)$ evaluated for the real test data and generated images. B: Violin plot of $r(X_p)$ stratified by digits label in the test data. C: Top row: Real test images; Bottom row: Fake generated images; From left to right: images with $r(X)$ values from the smallest, closest to 1, and the largest; the range of $r(X)$ for the shown images is in each subplot title.}
    \label{fig:MNIST_DCGAN}
\end{figure}

Figure~\ref{fig:MNIST_VAE}(B) and Figure~\ref{fig:MNIST_DCGAN}(C) provides examples of real and fake images whose values $r(X)$ are closest to 0, 1, and 2, respectively. For VAE, because the ranges of both $r(X_p)$ and $r(X_q)$ are very narrow and non-overlapping, the left and middle panels are the same for the real sample (top row); the middle and right panels have overlapping instances in the fake sample (bottom row). Although the generation quality of VAE is not ideal, it can still produce a few realistic images. The fake VAE samples on the left panel (near 0) illustrate poor fidelity, whereas those on the right appear more realistic (near 2). For the DCGAN result in Figure~\ref{fig:MNIST_DCGAN}, the range of $r(X_p)$ and $r(X_q)$ almost fully overlaps. However, the RDR values indicate DCGAN has its own distributional bias compared to the real sample. The over-represented samples (left panel) in Figure~\ref{fig:MNIST_DCGAN}(C) exhibit more rounded, spatially diffused patterns, such as ``5" and ``8", across the 28×28 grid, whereas the under-represented samples (right panel) are either geometrically simpler such as digit ``1”, or more irregular-looking.

\subsection{CelebA-64}\label{subsec:celeb}
The CelebA data \citep{CelebA64}, a widely used benchmark with human celebrity face images, is on a much larger scale, with 162{,}770 training and 19{,}962 testing observations in the form of $3\times 64\times 64$ for the RGB images. Importantly, it also contains rich covariate information: 40 binary facial attributes, such as smiling, male, wearing a hat, eyeglasses, etc. The covariate information can turn the distributional differences measured by $r(X)$ into interpretable features. This aspect (Desiderata C) is useful when evaluating synthetically generated scientific data, such as microbiome relative abundance data, where it is important to ensure that the key attributes match the real sample. We tried evaluating the generators in the low-dimensional feature space, but the low-dimensional mapping tends to smooth out the differences, and the choice of such mapping largely influences the evaluation results; hence, we still keep the evaluation in the original space. We use a DCGAN-style neural network architecture as the ratio learner, with a bounded sigmoid function $(\alpha=0.1)$ as the output to ensure the range is in $[0,2]$, and the training loss has a smoother decay with a small $\alpha$. Supplementary Section~\ref{supp_sec:act_sensi} provides the sensitivity analysis result when using the bounded softplus function as the output activation.

\begin{figure}[ht!]
    \centering
    \begin{subfigure}{0.3\textwidth}
        \includegraphics[width=\linewidth]{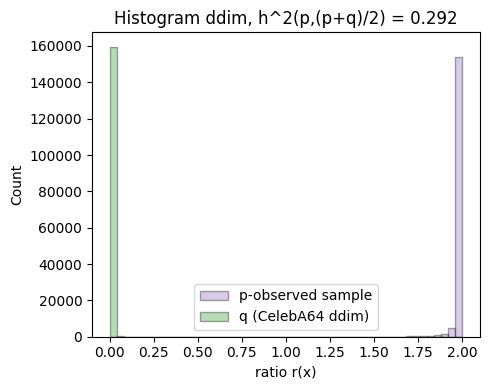}
        \caption{Histogram.}
        \label{fig:hist_ddim_alpha01}
    \end{subfigure}%
    \hfill
    \begin{subfigure}{0.27\textwidth}
        \centering
        \resizebox{\linewidth}{!}{%
        \begin{tabular}{lll}
        \toprule
       & $r(X_p)$ & $r(X_q)$ \\
       \hline
        length & 162770  & 160000  \\
        mean       & 1.99       & 0.00      \\
        std        & 0.05       & 0.03      \\
        min        & 0.10       & 0.00      \\
        q1         & 1.99       & 0.00      \\
        median     & 2.00       & 0.00      \\
        q3         & 2.00       & 0.00      \\
        max        & 2.00       & 1.97     \\
        \bottomrule
        \end{tabular}
        }
        \caption{Summary Statistics.}
        \label{tb:summary_ddim_alpha01}
    \end{subfigure}
    \hfill
    \begin{subfigure}{0.27\textwidth}
        \centering
        \resizebox{\linewidth}{!}{%
        \begin{tabular}{lcc}
        \toprule
        Feature               & Coef  & 95\% CI        \\ \hline
        \textbf{Wearing\_Hat}          & 0.63  & (0.61, 0.66) \\
        Bangs                 & 0.48  & (0.47, 0.50) \\
        Mouth\_Slightly\_Open & 0.44  & (0.43, 0.45) \\
        \textbf{5\_o\_Clock\_Shadow}   & -0.42 & (-0.44, -0.39) \\
        Blurry                & 0.39  & (0.36, 0.41) \\
        Pale\_Skin            & 0.25  & (0.22, 0.28) \\
        Black\_Hair           & -0.24 & (-0.25, -0.22) \\
        Bald                  & 0.22  & (0.18, 0.26) \\
        Sideburns             & -0.21 & (-0.24, -0.18) \\
        Rosy\_Cheeks          & -0.21 & (-0.23, -0.18)\\
        \bottomrule
        \end{tabular}
        }
        \caption{Linear regression of $\log((2-r(X_p))$ on 40 attributes.}
        \label{fig:ddim_log_regression}
        
    \end{subfigure}

    \vskip\baselineskip
    \begin{subfigure}{0.98\textwidth}
        \includegraphics[width=\linewidth]{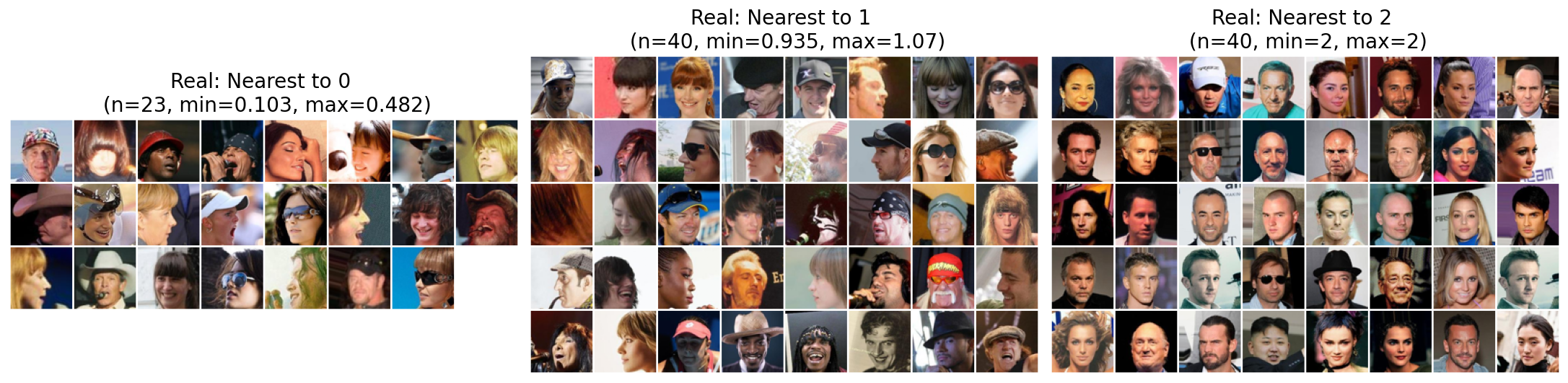}
        \caption{Real images.}
        \label{fig:CelebA_DDIM_real_img_alpha01}
    \end{subfigure}

    \vskip\baselineskip
    \begin{subfigure}{0.98\textwidth}
        \includegraphics[width=\linewidth]{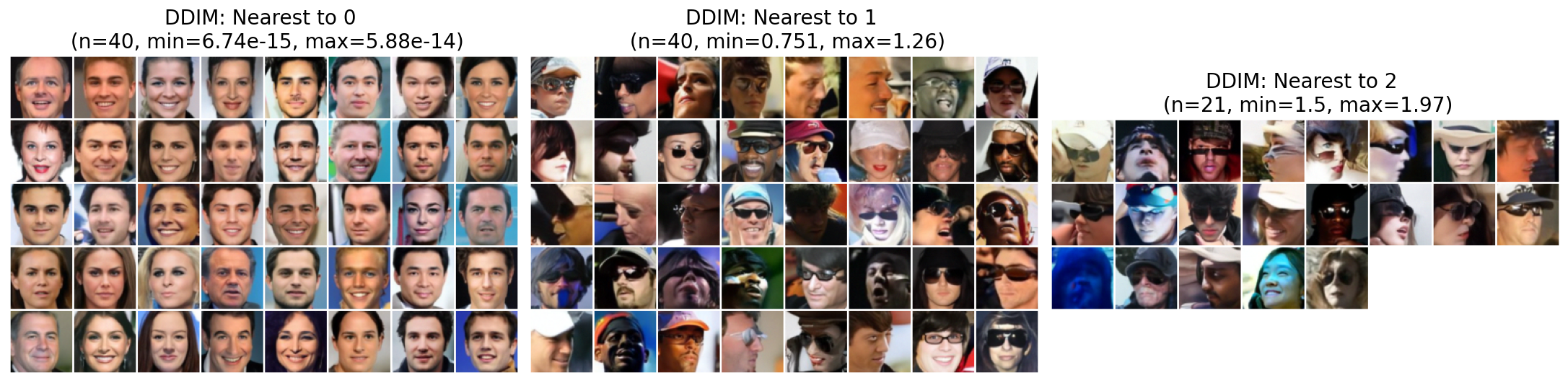}
        \caption{Fake images.}
        \label{fig:CelebA_DDIM_fake_img_alpha01}
    \end{subfigure}

    \caption{CelebA-64: DDIM evaluated on full training data.}
    \label{fig:CelebA_DDIM_alpha01}
\end{figure}

\textbf{DDIM.} Denoising diffusion models such as DDIM \citep{song2021ddim} are known to generate highly realistic images, but the sampling process is relatively slow. Therefore, in this experiment, when evaluating the pretrained DDIM model \citep{ddim_github}, we first generate 160,000 images and store them on disk. We then train the ratio function stochastically using batch updates from this pre-generated dataset. Figure~\ref{fig:CelebA_DDIM_alpha01} provides the analysis result on the DDIM generator. Even though the DDIM generated images in Figure~\ref{fig:CelebA_DDIM_fake_img} are visually realistic, but the vast distributional discrepancy shown in Figure~\ref{fig:hist_ddim} and Table~\ref{tb:summary_ddim} suggests otherwise. The DDIM generator has good coverage, since $r(X_p)$ and $r(X_q)$ have a similar range from near 0 to near 2. However, a vast majority of the DDIM-generated samples are highly concentrated around $r(X_q)$ near 0.  In fact, 159{,}719 out of 160{,}000 generated DDIM samples have $r(X_q)$ values smaller than the least $r(X_p)=0.1026$, indicating a large portion of possibly \textit{hallucinated} images. The left panel in Figure~\ref{fig:CelebA_DDIM_fake_img} shows such examples (more such examples in Supplementary Figure~\ref{fig:DDIM_900_outsupport}), and they are all centered, frontal or near-frontal portraits, with neutral to positive expressions, simple backgrounds, and the lighting appears relatively even and bright, without strong shadows. The overall structure of these faces is consistent, suggesting DDIM generated them from a tight distribution around a latent region representing a certain geometric average of the training data. To further study and interpret such distributional differences, we run a linear regression on the real training data, using $\log\br{2-r(X_p)}$ as the outcome, and the 40 attributes as the predictors. Table~\ref{tb:summary_ddim} further reveals some key features ranked by the largest absolute coefficients: for example, the attribute with the most negative effect is the \textit{5\_o\_Clock\_Shadow} (a feature indicating strong lighting in the image that cast shadows beneath the human face), meaning that it is difficult for DDIM to generate such images with shadows; on the other hand, it is easy for DDIM to generate realistic images with people wearing hat (feature with the most positive effect). These identified features can also be visually checked by comparing the middle and right panels in Figure~\ref{fig:CelebA_DDIM_real_img} and Figure~\ref{fig:CelebA_DDIM_fake_img}. 

This example highlights the strength of using RDR for distributional evaluation, particularly in high dimensions. For face images, distinctive differences between real and generated samples are readily perceived by human eyes. In contrast, for high-dimensional scientific data, such subtle discrepancies are often difficult to interpret directly. However, most scientific datasets are accompanied by rich covariate information, and the ratio-based regression approach leverages this context to make distributional differences more interpretable.

\subsection{American Gut Project}\label{subsec:microbiome}

For this analysis, we use the 2018 release of the American Gut Project \citep{AmericanGut} data, and the preprocessing is publicly available in our GitHub page \hidetext{(\url{https://github.com/yuliangxu/RDR})}. We use the relative abundance data derived from the taxa count data, such that the sum over all taxa for one sample equals 1. 
After preprocessing, we have 614 taxa (columns)  at the species level, and split the data into 80\% for training and 20\% for testing, resulting in training data of $n=10443$ and test data $n=2611$. We apply the ICFM (Independent Continuous Flow Matching) \citep{lipman2022flow, tong2023improving} as the generator on the logistic tree transformed space, where the relative abundance is aggregated through all levels in the phylogenetic tree and mapped to the real line using logit transform. We use the bounded sigmoid with $\alpha=0.1$ to have stably decaying training loss, and use the validation loss evaluated on the testing data to determine the early stop timing and avoid overfitting. The final results presented below are based on 856 training epochs.

\begin{figure}[ht!]
    \centering
    \begin{subfigure}{0.3\textwidth}
        \includegraphics[width=\linewidth]{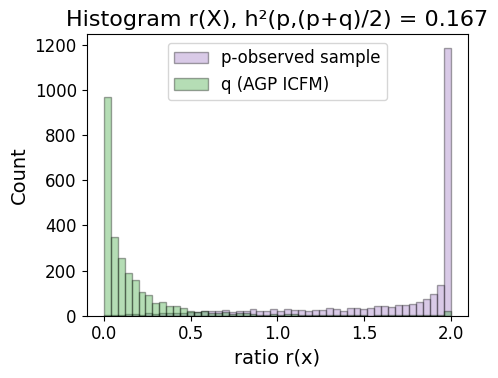}
        \caption{Histogram.}
        \label{fig:hist_icfm}
    \end{subfigure}%
    \hfill
    \begin{subfigure}{0.25\textwidth}
        \centering
        \resizebox{\linewidth}{!}{%
        \begin{tabular}{lll}
        \toprule
       & $r(X_p)$ & $r(X_q)$ \\
       \hline
        length & 2611    & 2611    \\
        mean   & 1.64       & 0.20       \\
        std    & 0.50       & 0.33       \\
        min    & 0.01       & 0.00       \\
        q1     & 1.41       & 0.02       \\
        median & 1.93       & 0.08       \\
        q3     & 2.00       & 0.21       \\
        max    & 2.00       & 2.00    \\
        \bottomrule
        \end{tabular}
        }
        \caption{Summary Statistics.}
        \label{tb:summary_icfm}
    \end{subfigure}
    \hfill
    \begin{subfigure}{0.35\textwidth}
        \includegraphics[width=\linewidth]{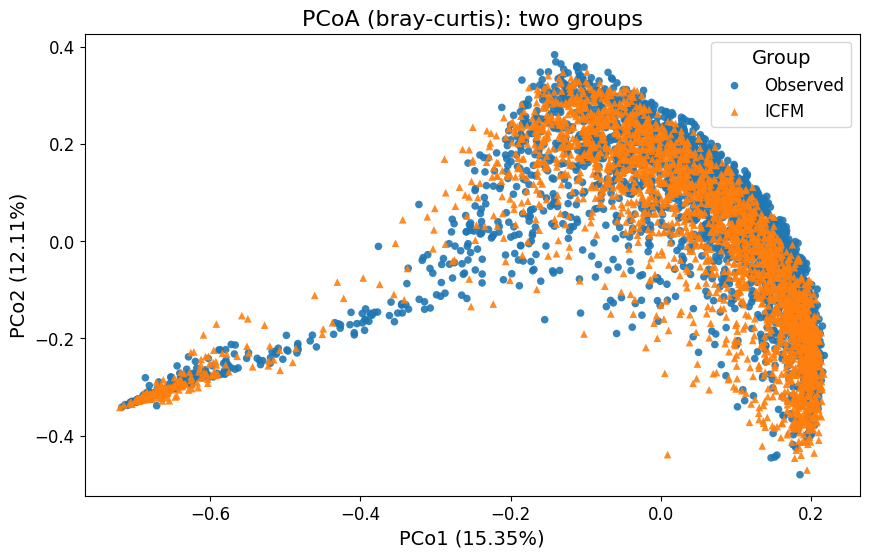}
        \caption{PCoA plot with Bray-Curtis metric.}
        \label{fig:icfm_pcoa}
    \end{subfigure}%

    \vskip\baselineskip
    \begin{subfigure}{0.98\textwidth}
        \includegraphics[width=\linewidth]{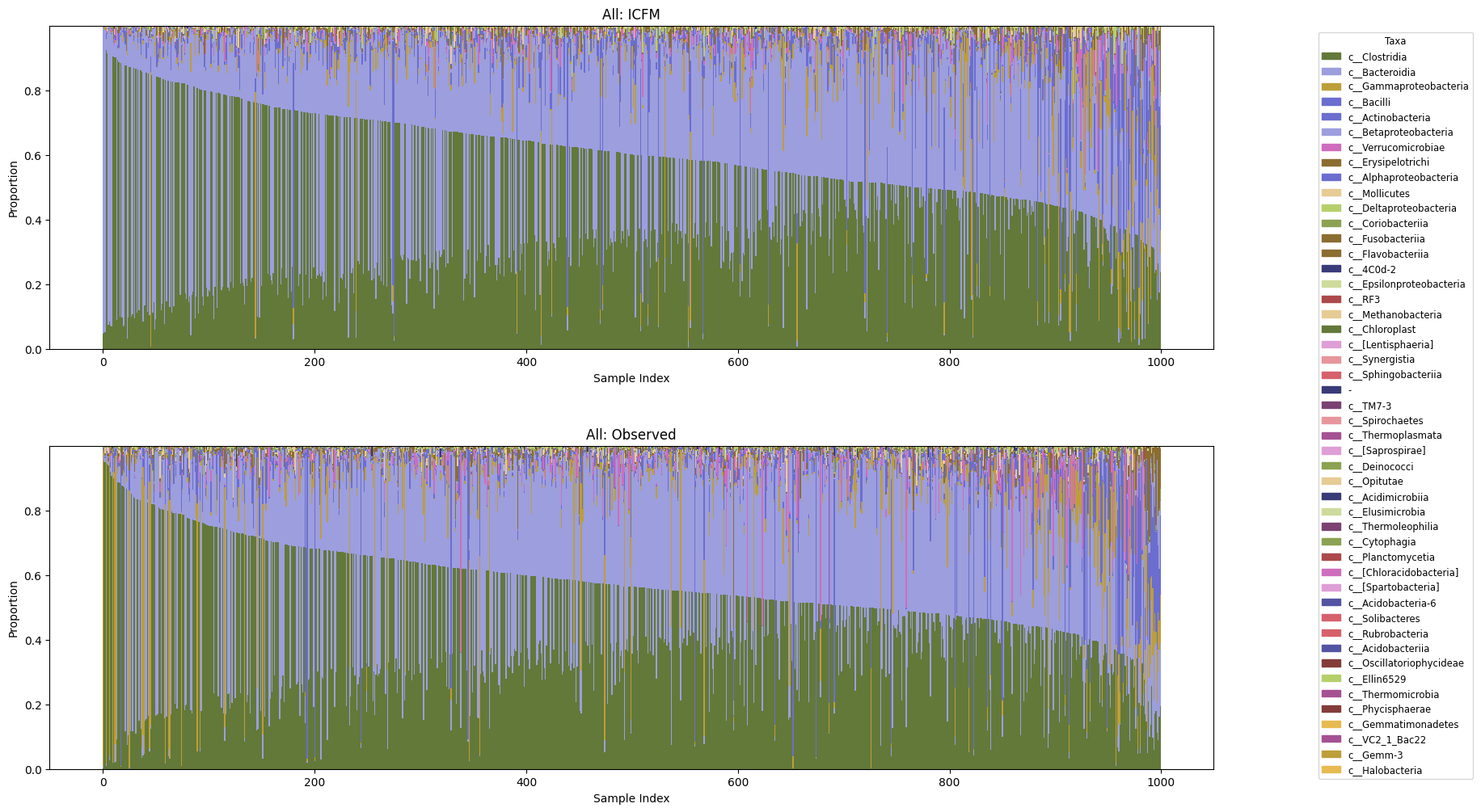}
        \caption{Stacked bar plot comparison on relative abundance.}
        \label{fig:agp_icfm_stacked_bar}
    \end{subfigure}
    \caption{American Gut: ICFM evaluated on testing data.}
    \label{fig:agp_icfm}
\end{figure}

Figure~\ref{fig:agp_icfm_stacked_bar} shows the class-level stacked bar plot in the simplex space of the generated sample (top) and the real sample (bottom). Each column in the stacked bar plot represents one sample, and the y-axis is the proportional breakdown in each class-level relative abundance for better visualization. The samples are ordered by the largest abundance of taxa. Different from the previous CelebA example, where we can visually check the difference between the generated sample and the real sample, visually evaluating the generation quality of the relative abundance data becomes impossible, even with a much smaller dimensionality. Some common visualizations include the alpha diversities to compare within-sample diversity  (Figure~\ref{fig:agp_icfm_alpha_diversity}) and beta diversities such as the PCoA scatter plot using Bray-Curtis metric (Figure~\ref{fig:icfm_pcoa}) to compare between-sample diversity. The ICFM generator performs well in all of these visual comparisons.

To evaluate the distributional difference, we train a ratio function $r(x)$ using the training data and an equal-sized ICFM-generated sample using a simple 4-layer MLP with ReLU activations on the species-level simplex space. Figure~\ref{fig:hist_icfm} shows the histogram of $r(X)$ evaluated on the test sample and another equal-sized ICFM-generated sample, and Table~\ref{tb:summary_icfm} provides the summary statistics. We can see that ICFM still has good coverage based on the range of $r(X)$ in the summary statistics table, but the majority of the generated samples are concentrated in an area outside the support of the observed data. 

To have a deeper understanding of this distributional difference, in particular, which taxa are over-represented/under-represented in the generated sample, we conduct an association test between each level of taxa abundance and the learned ratio. For the association test, we first transform the relative abundance samples using the central log ratio (CLR) transformation, then stack the CLR-transformed real and generated samples, as well as the learned ratio for both samples. At each taxonomic level, we aggregate the relative abundance to that level, and run a Spearman correlation test between the learned ratio and the CLR-transformed samples. Figure~\ref{fig:icfm_sunburst} shows the correlation coefficients at each taxonomic level. If the correlation coefficient for taxon $k$ is positive (red), this means that when $r(X)$ is higher, taxon $k$ increases relative to the geometric mean of its level, and is more abundant in the real sample than the generated sample. In Figure~\ref{fig:icfm_sunburst},  the size of each node (circle sector) reflects the number of descendant taxa contained within that lineage. Based on Figure~\ref{fig:icfm_sunburst}, we can see the tendency that the ICFM-generated sample tends to over-represent taxa with few descendants, whereas for large classes such as \textit{Clostridia}, ICFM tends to have less representation (this is also hinted by Table~\ref{tab:class_association}). In a way, similar to the CelebA result, ICFM as a generator tends to concentrate its mass towards a more \textit{averaged} area in the simplex.

\section{Discussion}\label{sec:discussion}

In this work, we propose a distributional evaluation framework on generative models based on the relative density ratio. This evaluation method not only provides numerical comparisons but also reveals the distributional discrepancy between the generator and the true data-generating mechanism, with clear interpretations that enable downstream covariate analyses. The relative density ratio can be extended for $K$-sample comparison: given $P_1,\dots,P_K$, the proposed RDR in \eqref{eq:RDR} can be modified to $r_k(x):=\frac{p_k(x)}{\frac{1}{K}\sum_{k=1}^K p_k(x)}$. Our theoretical results reveal that the (relative) density ratio mapping preserves $\phi$-divergence of the original distributions. For neural network estimators, when the true ratio function is in the anisotropic Besov space where the smoothness in different directions is allowed to vary, the convergence rate of the ratio function only depends on the harmonic mean of the true function's smoothness and is free of the dimension, which provides a theoretical foundation that explains the good performance of neural-network-based ratio learner in high-dimensional spaces. We have conducted extensive numerical experiments on a wide range of modern popular generators in various frameworks, including VAE, GAN, diffusion model (DDIM), and continuous flow matching (ICFM). Through benchmark image datasets like MNIST and CelebA-64, and microbiome relative abundance data in the American Gut Project, we discover that even for generators that are able to provide realistic synthetic samples, on a distributional level, they all have various biases from the true data-generating mechanism. This phenomenon has been observed in the machine learning community, where synthetic data are used in retraining the generators to lower the cost of authentic data acquisition. However, due to the distributional shift of the synthetic data, these retrained generators tend to degrade to a local region, a phenomenon known as \textit{model collapse}. We anticipate that by offering a principled approach to evaluating distributional discrepancies between synthetic and authentic data, this work will provide valuable guidance for correcting generative model behavior and facilitating model retraining and other applications involving synthetic data.

\section*{Acknowledgement}
\hidetext{We are grateful to Naoki Awaya for his valuable comments and discussions. We thank Ganchao Wei for providing the code to preprocess the American Gut Data and the code to train the ICFM generator on the relative abundance data. This work is supported by NIGMS grant R01-GM135440.} We used ChatGPT-5 to make sentence-level edits in writing and to assist code debugging related to this project.

\section*{Data Availability}
Our GitHub page \hidetext{(\url{https://github.com/yuliangxu/RDR})} and Supplementary code contain all the reproducible workflows to recreate all the analysis results reported in this manuscript. All data used are openly available, and we have included the code to directly download the data online. We use pre-trained models for MNIST and CelebA-64 examples, and have included the URL to download the pretrained checkpoints in our reproducible workflow. 

\spacingset{1.1} 
\bibliographystyle{asa}
\bibliography{ref}

\appendix

\section{Proof of Proposition \ref{prop:RDR}}

\begin{proof}

To see that $r(x)\in [0,2]$, note that $r(x) = \frac{2}{1+q(x)/p(x)}$ where $q(x)/p(x)\in [0,\infty]$, hence $r(x)\in [0,2]$.

To see that the squared Hellinger distance $H^2(p,\frac{p+q}{2})\leq 1-\frac{1}{\sqrt{2}}$, note that
\begin{align*}
    H^2\sbr{p,\frac{p+q}{2}}&=\frac{1}{2}\int_\cX \sbr{\sqrt{p(x)} - \sqrt{\frac{p(x)+q(x)}{2}}}^2 \leb(\dd x)\\
    &=1-\int_\cX \sqrt{p(x)\frac{p(x)+q(x)}{2}} \leb(\dd x) \leq 1-\int_\cX \frac{p}{\sqrt{2}} \leb(\dd x) = 1-\frac{1}{\sqrt{2}}.
\end{align*}

\end{proof}

\section{Proofs for general estimators}

\subsection{Proof of Theorem \ref{thm:dpi}}
\begin{proof}
    According to the relative data processing inequality, $D_\phi(P\|Q) \geq D_\phi(F_{\#}P \| F_{\#}Q)$. By Theorem 3.3 in \cite{Liese2012-rl} equivalent conditions (A) and (C), the equality holds when the transformation $F$ is any sufficient statistic for the family of probabilities $\cP := \{P, Q\}$. 

    According to Theorem 6.12 in Chapter 1 in \cite{lehmann2006theory} (which requires $P$ and $Q$ have common support), the density ratio $g(X)=\frac{p(X)}{q(X)}$ is a minimal sufficient statistic for the family $\cP$. Hence $g$ is a maximizer for $D_\phi(F_{\#}P \| F_{\#}Q)$ among all other measurable maps $F:\cX \to \cY$. 

    Consequently, after any one-to-one mapping, $T\circ g$ is also sufficient for $\cP$, hence also a maximizer for $D_\phi(F_{\#}P \| F_{\#}Q)$.
\end{proof}

\subsection{Proof of Theorem \ref{thm:consistency}}

\begin{proof}

The difference $|\hat H^2-H^2(\bP\|\bQ)|$ can be decomposed into two terms:
\begin{align*}
    \cE_0(\cG):=H^2(\bP\|\bQ) - \sup_{g\in\cG}\br{1-
    \frac{1}{2}\int g^{-1/2} \dd\bP -\frac{1}{2}\int g^{1/2}\dd\bQ}, \\
    \cE_1(\cG):= \sup_{g\in\cG}\br{
    \frac{1}{2}\int g^{-1/2} \dd(\bP_n-\bP) +\frac{1}{2}\int g^{1/2}\dd(\bQ_m-\bQ)}.
\end{align*}
And $|\hat H^2-H^2(\bP\|\bQ)|\leq \cE_0(\cG)+\cE_1(\cG)$. Due to Assumption~\ref{asm:gen_consist_true}, the true density ratio $g_0$ exists almost surely in $\cG$,  and $\cE_0(\cG)\overset{\as}{=}0$.
\begin{align*}
    &\left| \int g^{-1/2} \dd(\bP_n-\bP) +\int g^{1/2}\dd(\bQ_m-\bQ) \right|\\
    \leq & \underbrace{\left| \int \sbr{g^{-1/2} - g_0^{-1/2}} \dd(\bP_n-\bP) +\int \sbr{g^{1/2}-g_0^{1/2}}\dd(\bQ_m-\bQ)\right|}_{v_n(g)} +  \underbrace{\left| \int g_0^{-1/2} \dd(\bP_n-\bP) +\int g_0^{1/2}\dd(\bQ_m-\bQ)\right|}_{w_n}.
\end{align*}
Note that $\int g_0^{1/2}\dd\bQ = \int g_0^{-1/2}\dd\bP =\int \sqrt{p_0q_0}\dd\leb$ where $\leb$ is the base measure. By H\"older inequality, $\int\sqrt{p_0q_0}\dd\leb \leq \|\sqrt{p_0}\|_{L^2}\|\sqrt{q_0}\|_{L^2} = \sqrt{\int p\dd\leb}\sqrt{\int q\dd\leb}=1 <\infty$.  By the strong law of large number, $\int g_0^{1/2}\dd\bQ_m \overset{\as}{\to}\int g_0^{1/2}\dd\bQ$ and $\int g_0^{-1/2}\dd\bP_n\overset{\as}{\to}\int g_0^{-1/2}\dd\bP$, and both limits are finite,  $w_n\overset{\as}{\to}0$. 

By Assumption~\ref{asm:g_L1_bound} and Assumption~\ref{asm:general_entropy}, and Theorem 3.7 in \cite{geer2000empirical}, 
\begin{align*}
    &\sup_{g\in\cG} \int \sbr{g^{-1/2} - g_0^{-1/2}} \dd(\bP_n-\bP)\overset{\as}{\to} 0,\\
    &\sup_{g\in\cG} \int \sbr{g^{1/2}-g_0^{1/2}}\dd(\bQ_m-\bQ)\overset{\as}{\to} 0.
\end{align*}
Hence $\sup_{g\in\cG} v_n(g)\overset{\as}{\to}0$. Consequently, $|\hat H^2-H^2(\bP\|\bQ)|\overset{\as}{\to}0$.

Next, to show the consistency of $\hat g_n$, we first define 
\begin{equation}\label{eq:pseudo_distance}
    d(g,g_0):=\int (g^{1/2}-g_0^{1/2})\dd\bQ + \int (g^{-1/2}-g_0^{-1/2})\dd\bP
\end{equation}
as a pseudo distance between $g$ and $g_0$. We want to upper bound $d(\hat g_n,g_0)$ by $\sup_{g\in\cG}v_n(g)$, and lower bound $d(g,g_0)$ by a proper metric.

To show the upper bound, since $\hat g_n=\arg\min_{g\in\cG} \br{\int g^{-1/2}\dd\bP_n + \int g^{1/2}\dd\bQ_m}$, note that $\int \hat g_n^{-1/2}\dd\bP_n +\int \hat g_n^{1/2}\dd\bQ_m\leq \int g_0^{-1/2}\dd\bP_n +\int g_0^{1/2}\dd\bQ_m$. Then
\begin{align*}
    d(\hat g_n,g_0) &= \int (\hat g_n^{1/2}-g_0^{1/2})\dd\bQ + \int (\hat g_n^{-1/2}-g_0^{-1/2})\dd\bP\\
    &\leq \int (\hat g_n^{1/2}-g_0^{1/2})\dd(\bQ-\bQ_m) + \int (\hat g_n^{-1/2}-g_0^{-1/2})\dd(\bP-\bP_n) \leq \sup_{g\in\cG} v_n(g) \overset{\as}{\to} 0. \numberthis \label{eq:sup_v_bound}
\end{align*}

To show the lower bound, note that for any $g\in\cG$,
\begin{align*}
    d(g,g_0)&= \int (g^{1/2}-g_0^{1/2})\dd\bQ + \int ( g^{-1/2}-g_0^{-1/2})\dd\bP\\
    &=\int \frac{\sbr{\sqrt{g}-\sqrt{g_0}}^2}{\sqrt{g}} \dd\bQ.
\end{align*}
Given Assumption~\ref{asm:g_L1_bound} where $M_0 = \int G_0(x)\dd\bQ  =\int \sup_{g\in \cG}|g^{1/2}|\dd\bQ <\infty$, and H\"older inequality,
\begin{align*}
    \sbr{\int |\sqrt{g} - \sqrt{g_0}|\dd\bQ}^2 & = \sbr{ \int g^{1/4} \br{\frac{\sbr{\sqrt{g} - \sqrt{g_0}}^2}{\sqrt{g}}}^{1/2} \dd\bQ}^2\\
    &\leq \int g^{1/2}\dd\bQ \cdot \int \frac{\sbr{\sqrt{g}-\sqrt{g_0}}^2}{\sqrt{g}} \dd\bQ,\\
    \frac{1}{M_0}\sbr{\int |\sqrt{g} - \sqrt{g_0}|\dd\bQ}^2 &\leq d(g,g_0). \numberthis \label{eq:L1_lower_bound}
\end{align*}
This implies $\|\sqrt{\hat g_n}-\sqrt{g_0}\|_{L^1(\bQ)}\leq \sqrt{M_0d(\hat g_n,g_0)}\overset{\as}{\to}0$.

To show the last part, notice that $d(g,g_0) = \int \sqrt{g} \sbr{\frac{1}{\sqrt{g_0}}-\frac{1}{\sqrt{g}}}^2\dd\bP$. Similarly, denote $M_1=\int G_1\dd\bP<\infty$ in Assumption~\ref{asm:g_L1_bound},
\begin{align*}
    \sbr{\int \absbr{\frac{1}{\sqrt{g_0}}-\frac{1}{\sqrt{g}}}\dd\bP}^2 &= \sbr{\int g^{-1/4}\sbr{\sqrt{g}\absbr{\frac{1}{\sqrt{g_0}}-\frac{1}{\sqrt{g}}}}^{1/2} \dd\bP}^2\\
    &\leq \int g^{-1/2}\dd\bP\int \sqrt{g}\sbr{\frac{1}{\sqrt{g_0}}-\frac{1}{\sqrt{g}}}^2\dd\bP,\\
    \|\hat g_n^{-1/2} - g_n^{-1/2}\|_{L^1(\bP)}\leq \sqrt{M_1 d(\hat g_n,g_0)}\overset{\as}{\to}0.
\end{align*}

\end{proof}

\subsection{Proof of Lemma \ref{lem:entropy_bound}}
\begin{proof}

This result is mentioned in \citet{Nguyen2010-xb}. We give a formal proof here.
By Theorem 2.4.1 in \cite{van1996weak}, the bracketing number $N_{[]}(\epsilon, \cF, L^1(\bP))<\infty$ implies $\cF$ is Glivenko-Cantelli, i.e. $\sup_{f\in\cF} |\bP_n f-\bP f|\to 0$ a.s.

Let $N_{[]}(\epsilon, \cF, L^1(\bP)) =M<\infty$. By the definition of bracketing number, there exist $M$ pairs $\br{f^L_m, f^U_m}_{m=1}^M$ such that $\|f^L_m-f^U_m\|_{L^1(\bP)}\leq \epsilon$ and for any $f\in\cF$, we can find a pair indexed by $m$ such that $f^L_m\leq f\leq f^U_m$. Hence for any $f\in\cF$, there exists $m$ such that $f^L_m\leq f\leq f^U_m$,
\begin{align*}
    \|f-f_m^L\|_{L^1(\bP_n)} - \|f-f_m^L\|_{L^1(\bP)}
    &= \bP_n (f-f_m^L) - \bP (f-f_m^L) = (\bP_n-\bP)f- (\bP_n-\bP)f_m^L,\\
     \sup_{f\in\cF}\big|\|f-f_m^L\|_{L^1(\bP_n)} - \|f-f_m^L\|_{L^1(\bP)} \big|&
     \leq \sup_{f\in\cF} |(\bP_n-\bP)f| + \max_{1\leq m\leq M} |(\bP_n-\bP)f_m^L| \overset{\as}{\to} 0,
\end{align*}
where $\sup_{f\in\cF} |(\bP_n-\bP)f|\to 0$ a.s. is because $\cF$ is Glivenko-Cantelli, and $\max_{1\leq m\leq M} |(\bP_n-\bP)f_m^L|\to 0$ a.s. is by applying the strong law of large number $M$ times. The LHS $\sup_{f\in\cF}\big|\|f-f_m^L\|_{L^1(\bP_n)} - \|f-f_m^L\|_{L^1(\bP)} \big| \to 0$ a.s. implies that there exists $\Omega_0, N_\epsilon$ such that $\bP(\Omega_0)=1$, for any $n>N_\epsilon$ and any $\omega\in\Omega_0$,
\[\sup_{f\in\cF}\big|\|f-f_m^L\|_{L^1(\bP_n)} - \|f-f_m^L\|_{L^1(\bP)} \big| < \epsilon. \]

Since $\|f-f_m^L\|_{L^1(\bP)} \leq \epsilon$, then $\|f-f_m^L\|_{L^1(\bP_n)}\leq 2\epsilon$ for any $\omega\in\Omega_0$. Hence $\br{f_m^L}_{m=1}^M$ form an $2\epsilon$-covering set of cardinality $M$ with metric $L^1(\bP_n)$ a.s. Therefore the smallest covering set $N(2\epsilon,\cF,L^1(\bP_n))\leq N_{[]}(\epsilon, \cF, L^1(\bP))$ a.s. when $n>N_\epsilon$. Since $\epsilon$ is arbitrary, $N(\epsilon,\cF,L^1(\bP_n)) <\infty$ a.s. and also bounded in probability, and $\frac{1}{n}N(\epsilon,\cF,L^1(\bP_n))\overset{P}{\to}0$.
\end{proof}

\subsection{Proof of Theorem \ref{thm:convergence_rate}}

\begin{proof}
    This proof follows Theorem 3.2.5 in \cite{van1996weak}. 

    First, we introduce some notations to relate the M-estimator theory in our context. Denote the objective function in \eqref{eq:hellinger_objective} as
    \begin{equation}\label{eq:M_obj}
        \bM(g):=1-\frac{1}{2}\int g^{-\frac{1}{2}} \dd\bP - \frac{1}{2}\int g^{\frac{1}{2}}\dd\bQ,
    \end{equation}
    and the corresponding empirical process as $\bM_n(g):=1-\frac{1}{2}\int g^{-\frac{1}{2}} \dd\bP_n - \frac{1}{2}\int g^{\frac{1}{2}}\dd\bQ_m$.

    To show the convergence rate, we need to have a metric in the function space $\cG$. In the previous Theorem \ref{thm:consistency}, the $L^1(\bQ)$ distance is used. For the convergence rate result, because most of the existing literature on the entropy number condition is stated in terms of $L^2$ norms, we need a $L^2$ type metric to verify the entropy number conditions. Hence we use a different lower bounding metric here than in \eqref{eq:L1_lower_bound}, given $\sup_{g\in\cG}|g| < M_u$ in Assumption~\ref{asm:g_bound},
    \begin{equation}\label{eq:d_less_h}
        d(g,g_0)\geq \frac{1}{\sqrt{M_u}} \int \sbr{\sqrt{g}-\sqrt{g_0}}^2 \dd\bQ =\frac{1}{\sqrt{M_u}} h_\bQ^2(g,g_0).
    \end{equation}

    Now we need to verify the following four conditions from Theorem 3.2.5 in \cite{van1996weak}.
    
    \noindent\textbf{Conditions to be verified:}
    \begin{enumerate}
        \item $\bM(g) - \bM(g_0) \lesssim - h_\bQ^2(g,g_0)$;
        \item For every $n$ and sufficiently small $\delta$,
        \begin{equation}\label{eq:condition2}
            \bE \sup_{h_\bQ(g,g_0)<\delta} \left|
        (\bM_n-\bM)(g) - (\bM_n-\bM)(g_0)
        \right| \lesssim \frac{\phi_n(\delta)}{\sqrt{n}}
        \end{equation}
        such that $\frac{\phi_n(\delta)}{\delta^\alpha}$ is decreasing  for some $\alpha<2$;
        \item Let $r_n^2\phi_n\sbr{\frac{1}{r_n}}\leq \sqrt{n}$ for every $n$, and verify that $\hat g_n$ satisfies $\bM_n(\hat g_n) \geq \bM_n(g_0) - O_p(r_n^{-2})$;
        \item $h_\bQ(\hat g_n,g_0)\overset{p}{\to}0$.\footnote{The original Theorem 3.2.5 in \cite{van1996weak} is stated with outer probability in Condition 4 and outer expectation in Condition 2. Here, we omit the discussion on the measurability of $\hat g_n$.}
    \end{enumerate}
With these conditions, one can conclude that $r_n h_\bQ(\hat g_n, g_0) = O_p(1)$.

Note that in \eqref{eq:pseudo_distance}, $d(g,g_0) = -\bM(g) + \bM(g_0)$, and Condition 1 is verified by \eqref{eq:d_less_h}. By the inequalities \eqref{eq:d_less_h} and \eqref{eq:sup_v_bound}, $h_\bQ^2(\hat g_n,g_0)\lesssim d(\hat g_n,g_0) \overset{p}{\to} 0$, hence Condition 4 is verified. The inequality in Condition 3 is satisfied, given that $\hat g_n$ is the maximizer of $\bM_n(g)$ in $\cG$ and $g_0\in \cG$. Now, we focus on finding $\phi_n(\delta)$ to satisfy Condition 2 and $r_n^2\phi_n(\frac{1}{r_n})\leq \sqrt{n}$. 

We adopt the following empirical process notations 
$$\bG_n^P f:=\sqrt{n}\int f \dd \sbr{\bP_n - \bP},\quad\bG_m^Q f:=\sqrt{m}\int f \dd \sbr{\bQ_m - \bQ},$$ and denote $\|\cG_n\|_{\cF} := \sup_{f\in\cF}|\bG_n f|$. The inequality \eqref{eq:condition2} can be written as 
\[\bE \sup_{h_\bQ(g,g_0)<\delta} \left|  \bG_n^P(g^{-1/2}-g_0^{-1/2}) +  \bG_m^Q(g^{1/2}-g_0^{1/2}) \right|\lesssim \phi_n(\delta).\]
Hence we only need to find $\phi_n(\delta)$ such that 
\begin{equation*}
    \text{(A) } \bE \|\bG_n^P\|_{\cF_\delta^P} \lesssim \phi_n(\delta), \quad \text{and (B) } \bE \|\bG_m^Q\|_{\cF_\delta^Q} \lesssim \phi_m(\delta),
\end{equation*}
where $\cF_\delta^P,\cF_\delta^Q$ are defined as in \eqref{eq:def_F}.

We use Lemma 3.4.2 in \cite{van1996weak} to show the entropy condition. We first show (B), and the proof for (A) is similar. By the definition \eqref{eq:def_F}, for any $f\in \cF_\delta^Q$, $\bQ f^2<\delta^2$, and Assumption~\ref{asm:g_bound} implies there exists a positive constant $M_{\cF^Q}$ such that $\sup_{g\in\cG}|g^{1/2} - g_0^{1/2}|\leq M_{\cF^Q}<\infty$. By Lemma 3.4.2 in \cite{van1996weak}
\begin{equation}\label{eq:entropy_bound_q}
    \bE \|\bG_m^Q\|_{\cF_\delta^Q} \lesssim \tilde J_{[]}(\delta, \cF_\delta^Q, L^2(\bQ)) \sbr{1+\frac{\tilde J_{[]}(\delta, \cF_\delta^Q, L^2(\bQ))}{\delta^2\sqrt{m}}M_{\cF^Q}}.
\end{equation}

According to Assumption~\ref{asm:general_entropy_integral}, for any small $\delta<\eta$,
\begin{align*}
 \tilde J_{[]}(\delta, \cF_\delta^Q, L^2(\bQ)) &\leq 
 \int_0^\delta \sqrt{1+\log N_{[]}(\epsilon,\cF_\eta^Q,L^2(\bQ))} \dd \epsilon
 \lesssim \int_0^{\delta}\sqrt{1+\epsilon^{-\gamma_q}}\dd\epsilon.   
\end{align*}

Because $\delta$ is allowed to be arbitrarily small, we can assume that $\delta<1$. Hence when $\epsilon<1$, $0<\gamma_q<2$, $1+\epsilon^{-\gamma_q}<2\epsilon^{-\gamma_q}$, $\int_0^\delta\sqrt{1+\epsilon^{-\gamma_q}}\dd\epsilon\leq \sqrt{2}\int_0^\delta \epsilon^{-\gamma_q/2}\dd\epsilon = \frac{\sqrt{2}}{1-\gamma_q/2}\delta^{-\gamma_q/2+1}$. When we plug in the upper bound $\tilde J_{[]}(\delta, \cF_\delta^Q, L^2(\bQ))\leq \frac{\sqrt{2}}{1-\gamma_q/2}\delta^{-\gamma_q/2+1}$ to \eqref{eq:entropy_bound_q}, $\bE \|\bG_m^Q\|_{\cF_\delta^Q} \lesssim \delta^{-\gamma_q/2+1}$.

The proof for (A) is similar, except that when verifying the second moment bound, for any $f\in\cF_\delta^P$,
\begin{align*}
    \bP f^2&=\int \sbr{\frac{1}{\sqrt{g}} - \frac{1}{\sqrt{g_0}}}^2 \dd\bP = \int \frac{(\sqrt{g} - \sqrt{g_0})^2}{gg_0}\dd\bP \\
    &\overset{(*)}{\leq} M_l^2 \int (\sqrt{g}-\sqrt{g_0})^2 \dd\bP 
    \leq M_l^2 \int g_0 (\sqrt{g}-\sqrt{g_0})^2 \dd\bQ \\
    &\overset{(**)}{\leq} M_l^2 M_u \int (\sqrt{g}-\sqrt{g_0})^2 \dd\bQ \leq C_M^2\delta^2. \numberthis \label{eq:pf2_verify}
\end{align*}
Both inequalities $(*)$ and $(**)$ are due to Assumption~\ref{asm:g_bound}. $C_M = \sqrt{M_l^2 M_u}$ is a constant. Similarly, for small $\delta$ where $\delta C_M<1$, $\tilde J_{[]}(C_M\delta, \cF_\delta^P,L^2(\bP))\leq \frac{\sqrt{2}(C_M\delta)^{-\gamma_p/2+1}}{1-\gamma_p/2}$, and $\bE\|\bG_n^\bP\|_{\cF_\delta^\bP}\lesssim \delta^{-\gamma_p/2+1}$. 

Choose $\phi_n(\delta) =\delta^{-\gamma/2+1}$ where $\gamma = \min\br{\gamma_p,\gamma_q}$. The function $\phi_n(\delta)/\delta^\alpha$ is decreasing when we take some $\alpha$ in $(1,2)$. Now we can solve for $r_n$ from $r_n^2\phi_n(1/r_n)\leq \sqrt{n} \wedge \sqrt{m}$, and under Assumption~\ref{asm:sample_size_ratio}, $r_n\lesssim n^{\frac{1}{2+\gamma}}$.

\end{proof}

\section{Proofs for neural network sieved space}

Let $g_n$ be a point in $\cG_n$ ($g_n$ can be viewed as the projection of $g_0$ in $\cG_n$). The convergence rate $h_\bQ(\hat g_n, g_0) \leq h_\bQ(\hat g_n, g_n) + h_\bQ(g_n, g_0)$. The first term is the empirical process estimation error, and the second term is the neural network approximation error. 

To show the rate of the estimation error, we need to verify conditions in Theorem 3.4.1 in \cite{van1996weak}.

\noindent\textbf{Conditions to be verified:}
    
    For every $n$ and $\delta_n<\delta < \infty$,
    \begin{enumerate}
        \item $\sup_{\delta/2 < h_\bQ(g,g_n)\leq \delta} \bM(g) - \bM(g_n) \lesssim -\delta^2$;
        \item 
        \begin{equation}\label{eq:sieve_condition2}
            \bE \sup_{\delta/2 < h_\bQ(g,g_n)\leq \delta} \left|
        (\bM_n-\bM)(g) - (\bM_n-\bM)(g_n)
        \right| \lesssim \frac{\phi_n(\delta)}{\sqrt{n}}
        \end{equation}
        such that $\frac{\phi_n(\delta)}{\delta^\alpha}$ is decreasing on $(\delta_n,+\infty)$ for some $\alpha<2$;
        \item Let $r_n \lesssim \delta_n^{-1}$ satisfy $r_n^2\phi_n\sbr{\frac{1}{r_n}}\leq \sqrt{n}$ for every $n$.
        \item $\hat g_n$ takes values in $\cG_n$ and satisfies $\bM_n(\hat g_n) \geq \bM_n(g_n) - O_p(r_n^{-2})$.
    \end{enumerate}
With these conditions, one can conclude that $r_n h_\bQ(\hat g_n, g_n) = O_p(1)$. Here, $\delta_n$ is a multiple of $h_\bQ(g_n, g_0)$. Note that $h_\bQ(g_n, g_0)$ depends on the approximation rate in Proposition 2 of \cite{suzuki2021deep}. The sketch of the proof is as follows: first, we use Lemma~\ref{lem:sieve_cond1} to show Condition 1; then use Lemma~\ref{lem:nn_bracket_number} to compute the bracketing number of the sparse neural network space that will be used in the proof of Condition 2; finally, we will show Condition 2 to 4 in the proof of Theorem~\ref{thm:sieve_nn_rate}, and solve for the rate $r_n$ in Condition 3 where $r_n\asymp \delta_n^{-1}$, hence the final convergence rate will be the same order of $h_\bQ(\hat g_n, g_n)$ and $h_\bQ(g_n, g_0)$.

\begin{lemma}\label{lem:sieve_cond1}
Under Assumption~\ref{asm:g_bound},
    for $g_n\in \cG_n, g\in \cG$, and $g_0$ is the maximizer of $\bM(g)$, if $h_\bQ(g_n, g) \geq c_M h_\bQ(g_n, g_0)$, then $\bM(g) - \bM(g_n) \leq -\frac{1}{2\sqrt{M_u}} h^2_\bQ(g,g_n)$. $c_M$ is a constant where $c_M = \frac{1}{4M_lM_u}$.
\end{lemma}
\begin{proof}
    The intuition behind this proof is that, for $g_n$ in $\cG_n$, we want $g_n$ to be closer to the global optimizer $g_0$ than the arbitrary $g$ when verifying the first condition in Theorem 3.4.1 in \cite{van1996weak}. 

For the objective $\bM(g)$ defined in \eqref{eq:M_obj},
    \begin{align*}
        \bM(g) - \bM(g_n) &= \int p_0 (g_n^{-1/2} - g^{-1/2}) + q_0 (g_n^{1/2} - g^{1/2}) \dd\mu\\
        &= \int g_0\frac{\sqrt{g} - \sqrt{g_n}}{\sqrt{g_n g}} +(\sqrt{g_n} - \sqrt{g}) \dd\bQ\\
        &= \int (g_0-g_n+g_n)\frac{\sqrt{g} - \sqrt{g_n}}{\sqrt{g_n g}} +(\sqrt{g_n} - \sqrt{g}) \dd\bQ\\ 
        &= \int (g_0-g_n)\frac{\sqrt{g} - \sqrt{g_n}}{\sqrt{g_n g}} - \frac{(\sqrt{g_n}-\sqrt{g})^2}{\sqrt{g}} \dd\bQ.
    \end{align*}
Given Assumption~\ref{asm:g_bound}, $\int_\cX \sup_{g\in \cG_n}|g^{1/2}(x)|\dd\bQ\leq M_u^{1/2}$, the second term in the last equation can be bounded in the same way as in \eqref{eq:L1_lower_bound},
\begin{align*}
    -\int \frac{(\sqrt{g_n}-\sqrt{g})^2}{\sqrt{g}} \dd\bQ\leq  -\frac{1}{M_u^{1/2}} h^2_\bQ(g_n,g).
\end{align*}
For the first term, based on Assumption~\ref{asm:g_bound},
\begin{align*}
     \int &(g_0-g_n)\frac{\sqrt{g} - \sqrt{g_n}}{\sqrt{g_n g}}\dd\bQ  = \int (\sqrt{g_0}-\sqrt{g_n})(\sqrt{g_0}+\sqrt{g_n})\frac{\sqrt{g} - \sqrt{g_n}}{\sqrt{g_n g}}\dd\bQ  \\
     &\leq 2M_l\sqrt{M_u}  \int (\sqrt{g_0}-\sqrt{g_n}) (\sqrt{g}-\sqrt{g_n}) \dd\bQ \\
     &\leq 2M_l\sqrt{M_u} h_\bQ(g_0,g_n) h_\bQ(g,g_n).
\end{align*}
Now if we let $h_\bQ(g_0,g_n) \leq c_M h_\bQ(g,g_n)$, where $c_M = \frac{1}{4M_lM_u}$, then 
\begin{align*}
    \bM(g) - \bM(g_n) &\leq -\frac{1}{2\sqrt{M_u}} h^2_\bQ(g,g_n).
\end{align*}

\end{proof}

The next lemma bounds the bracketing number of a given neural network functional class.

\begin{lemma}\label{lem:nn_bracket_number}
     Let $H(x):=d\|x\|_\infty+2$ where $d$ is the dimension of the sample space $\cX$. Without loss of generality, assume that $UB\geq 2$, then for any norm $\|\cdot\|$ in the function space $\cF$,
     \begin{align*}
         N_{[]}(\epsilon,\Phi(L,U,S,B),\|\cdot\|)&\leq  \sbr{\frac{2B^{L}U^{L}Le\|H(x)\|(2UL+d)}{\delta S}}^S.
     \end{align*}
\end{lemma}

\begin{proof}
This proof partially follows the ideas in Lemma S.2 in \cite{kaji2023adversarial} and Lemma 6 in \cite{suzuki2021deep}. The key step is to verify the conditions for Theorem 2.7.11 in \cite{van1996weak}.

Denote $f(x) = f^{(L)}\circ\dots\circ f^{(1)}$ as a function in $\Phi(L,U,S,B)$, indexed by the neural network parameters $\bftheta = \br{W^{(l)},b^{(l)}}_{l=1}^L$, and let $\tilde f(x) = \tilde f^{(L)}\circ\dots\circ \tilde f^{(1)}$ be another function in $\Phi(L,U,S,B)$, indexed by the neural network parameters $\tilde \bftheta = \br{\tilde W^{(l)},\tilde b^{(l)}}_{l=1}^L$. We first show that $|f(x)-\tilde f(x)|\leq Kd_\theta(\bftheta,\tilde\bftheta)F(x)$ where $d_\theta$ is a metric on the vector space of $\bftheta$. $K$ is a constant that only depends on $(L,U,S,B)$, and will be specified later.

For any $l\leq L-1$, note that the activation function $\sigma$ is 1-Lipschitz and $\sigma(0)=0$,
\begin{align*}
    \nbr{f^{(l)}\circ\dots\circ f^{(1)}}_\infty
    &\leq \nbr{\sbr{W^{(l)}(\cdot) + b^{(l)}}\circ\dots f^{(1)}}_\infty\\
    &\leq \max_i \nbr{W_{i\cdot}^{(l)}}_1 \nbr{f^{(l-1)}\circ\dots \circ f^{(1)}}_\infty + \nbr{b^{(l)}}_\infty\\
    &\leq UB \nbr{f^{(l-1)}\circ\dots \circ f^{(1)}}_\infty + B\\
    &\leq (UB)^{l-1}\nbr{f^{(1)}}_\infty + B\sum_{i=0}^{l-2}(UB)^i.
\end{align*}
Here, $\sum_{i=0}^{l-2}(UB)^i = \frac{1-(UB)^{l-1}}{1-UB}\leq (UB)^{l-1}$ given $UB\geq 2$. $\|f^{(1)}\|_\infty \leq Bd\|x\|_\infty + B$. Hence
\begin{align*}
     \|f^{(l)}\circ\dots\circ f^{(1)}\|_\infty &\leq (UB)^{l-1}\sbr{ Bd\|x\|_\infty + B} + B(UB)^{l-1}\\
     &= U^{l-1}B^{l}(d\|x\|_\infty + 2). 
\end{align*}
Now, let $H(x) = d\|x\|_\infty+2$. In addition, denote $d^{(l)}_w := \|W^{(l)}-\tilde W^{(l)}\|_\infty$, $d^{(l)}_b:=\|b^{(l)}-\tilde b^{(l)}\|_\infty$, $d^{(\infty)}_w:=\max_{l}d^{(l)}_w$, $d^{(\infty)}_b:=\max_{l}d^{(l)}_b$,
\begin{align*}
    |f(x)-\tilde f(x)| \leq &
     \bigg|\br{(W^{(L)}(\cdot)+b^{(L)})\circ f^{(L-1)}\circ\dots\circ f^{(1)}} - \br{(\tilde W^{(L)}(\cdot)+\tilde b^{(L)})\circ \tilde f^{(L-1)}\circ\dots\circ \tilde f^{(1)}} \bigg| \\
     = & \bigg|\br{(W^{(L)}(\cdot)+b^{(L)})\circ f^{(L-1)}\circ\dots\circ f^{(1)}} - \br{(\tilde W^{(L)}(\cdot)+\tilde b^{(L)})\circ f^{(L-1)}\circ\dots\circ f^{(1)}} + \\
     & \br{(\tilde W^{(L)}(\cdot)+\tilde b^{(L)})\circ f^{(L-1)}\circ\dots\circ f^{(1)}} - \br{(\tilde W^{(L)}(\cdot)+\tilde b^{(L)})\circ \tilde f^{(L-1)}\circ\dots\circ \tilde f^{(1)}}\bigg|\\
    \leq & Ud_w^{(L)} \|f^{(L-1)}\circ\dots\circ f^{(1)}\|_\infty + d_b^{(L)} + UB\|f^{(L-1)}\circ\dots\circ f^{(1)} - \tilde f^{(L-1)}\circ\dots\circ \tilde f^{(1)}\|_\infty\\
    \leq & (UB)^{L-1}H(x)d_w^{(L)} + d_b^{(L)}  + \\
    & UB\br{(UB)^{L-2}H(x)d_w^{(L-1)} + d_b^{(L-1)} + UB \|f^{(L-2)}\circ\dots\circ f^{(1)} - \tilde f^{(L-2)}\circ\dots\circ \tilde f^{(1)}\|_\infty}\\
    \leq & (L-1)(UB)^{L-1}H(x) d_w^{(\infty)} + d_b^{(\infty)}\sum_{l=0}^{L-2}(UB)^l + (UB)^{L-1}\|f^{(1)}-\tilde f^{(1)}\|_\infty.
\end{align*}
Note $\|f^{(1)}-\tilde f^{(1)}\|_\infty\leq d \times d_w^{(1)} \|x\|_\infty+d_b^{(1)}$, where $x\in \R^d$,
\begin{align*}
    |f(x)-\tilde f(x)|
    \leq & L(UB)^{L-1}H(x)d_w^{(\infty)} + (UB)^{L-1}d_b^{(\infty)}.
\end{align*}
Let $d(\bftheta,\tilde\bftheta)=\|\bftheta-\tilde\bftheta\|_\infty \geq \max\br{d_w^{(\infty)},d_b^{(\infty)}}$, then 
\begin{align*}
    |f(x)-\tilde f(x)|
    \leq & L(UB)^{L-1}\sbr{d\|x\|_\infty+2}d(\bftheta,\tilde\bftheta).
\end{align*}

Since all elements in $\bftheta$ is constrained within $[-B,B]$, the space of all possible $\bftheta$ equipped with $\|\cdot\|_\infty$ is a hypercube $[-B,B]^S$ given a fixed sparsity structure. Given that the input dimension is $d$ and the output dimension is $1$, there are a total of $$\sbr{\begin{array}{c}
     U^2(L-2)+Ud+U+U(L-1)+1  \\
      S
\end{array}}$$ such sparsity structures. Note that $\sbr{\begin{array}{c}
     n  \\
     S 
\end{array}}$ can be bounded by $n^S/S!$, and using Stirling’s inequality $S!\geq (S/e)^S$, $\sbr{\begin{array}{c}
     n  \\
     S 
\end{array}}\leq \sbr{en/S}^S$. In our case, $n$ is $U^2(L-2)+Ud+U+U(L-1)+1\leq U^2L+U(L+d)\leq U(2UL+d)$, given that $L,d\geq1$. Hence, the total number of sparsity structures can be bounded by $(eU(2UL+d)/S)^S$.

Denote $F(x):= L(UB)^{L-1}H(x) = L(UB)^{L-1}\sbr{d\|x\|_\infty+2}$. Also, note that using the vector infinity norm and radius $\epsilon$, the covering number of $[-B,B]^S$ is $(B/\epsilon)^S$.
According to Theorem 2.7.11 in \cite{van1996weak}, then  
\begin{align*}
    N_{[]}\sbr{2\epsilon \|F\|,\Phi(L,U,S,B),\|\cdot\|} &\leq \sbr{\frac{B}{\epsilon}\frac{eU(2UL+d)}{S}}^S.
\end{align*}

Hence,
\begin{align*}
    N_{[]}\sbr{\delta,\Phi(L,U,S,B),\|\cdot\|} &\leq \sbr{\frac{2B^{L}U^{L}Le\|H(x)\|(2UL+d)}{\delta S}}^S.
\end{align*}
\end{proof}

\subsection{Proof of Theorem~\ref{thm:sieve_nn_rate}}

\begin{proof}
We first show the proof of condition 2 and the inequality \eqref{eq:sieve_condition2}. Based on Assumption~\ref{asm:sample_size_ratio} $m=\rho n$, for simplicity, we just take $m=n$ in the following proof.

We again use the empirical process notations $\bG_n^P f = \sqrt{n}\int f \sbr{\dd\bP_n - \dd\bP}$.
    The LHS of equation \eqref{eq:sieve_condition2} can be written as 
    \begin{align*}
        \sqrt{n}\LHS  &\leq  \bE \sup_{\delta/2<h_\bQ(g,g_n)<\delta}\left| 
        \bG_n^P\sbr{g_n^{-1/2} - g^{-1/2}}
        \right| + \bE \sup_{\delta/2<h_\bQ(g,g_n)<\delta}\left| 
        \bG_n^Q\sbr{g_n^{1/2} - g^{1/2}}
        \right| \\
        &\leq  \bE \|\bG_n^P\|_{\cF_{\delta,n}^P} + \bE \|\bG_n^Q\|_{\cF_{\delta,n}^Q},
    \end{align*}
    where for any fixed $g_n\in\cG_n$ ($g_n$ is usually taken to be the projection of $g_0$ onto the sieve space $\cG_n$),
    \begin{align*}
        \cF_{\delta,n}^P &=\br{g^{-1/2} - g_n^{-1/2}: g\in\cG_n, h_\bQ(g_n,g)<\delta}, \quad
        \cF_{\delta,n}^Q =\br{g^{1/2} - g_n^{1/2}: g\in\cG_n, h_\bQ(g_n,g)<\delta}.
    \end{align*}

We first show how to find $\phi_n(\delta)$ that satisfy $ \bE \|\bG_n^Q\|_{\cF_{\delta,n}^Q}\lesssim \phi_n(\delta)$, and $ \bE \|\bG_n^P\|_{\cF_{\delta,n}^P}\lesssim \phi_n(\delta)$ follows similarly. 

By Lemma 3.4.2 in \cite{van1996weak}, if there exists $M$ such that $\|g_n^{1/2} - g^{1/2}\|_\infty \leq M$, and $\bQ\br{\sbr{g^{1/2} - g_n^{1/2}}^2}\lesssim \delta^2$, then 
\begin{align*}
    \bE \|\bG_n^Q\|_{\cF_{\delta,n}^Q} &\leq \tilde J_{[]}\sbr{\delta, \cF_{\delta,n}^Q, L^2(\bQ)} \sbr{1+\frac{\tilde J_{[]}\sbr{\delta, \cF_{\delta,n}^Q, L^2(\bQ)}}{\delta^2\sqrt{n}}M}.
\end{align*}
By Assumption~\ref{asm:g_bound}, such $M$ exists and is a constant that only depends on $M_l,M_u$, and the second moment condition is satisfied in the definition of $\cF_{\delta,n}^Q$.
We want to show that for some constant $C_g>0$, $N_{[]}\sbr{C_g\delta,\cF_{\delta,n}^Q,L^2(\bQ)}\leq N_{[]}\sbr{\delta,\cG_n,L^2(Q)}$, and use the entropy bound for $\cG_n$ that we have obtained in Lemma~\ref{lem:nn_bracket_number}. 

Denote $\cN_{[]}\sbr{\delta,\cG_n,L^2(\bQ)}$ as the minimal set of all bracketing pairs for $\cG_n$, with a total of $N_{[]}\sbr{\delta,\cG_n,L^2(Q)}$ pairs. For any $g\in\cG_n$, there exists a pair $\br{g^L,g^U}$ in $\cN_{[]}\sbr{\delta,\cG_n,L^2(Q)}$. We need to show that $\br{\sqrt{g^L}-\sqrt{g_n},\sqrt{g^U} - \sqrt{g_n}}$ forms a bracketing pair in $\cN_{[]}\sbr{C_g\delta,\cF_{\delta,n}^Q,L^2(Q)}$ for some $C_g$. To see this, firstly, $g^L\leq g\leq g^U$ implies $\sqrt{g^L} - \sqrt{g_n}\leq \sqrt{g} - \sqrt{g_n}\leq \sqrt{g^U} -\sqrt{g_n}$. Secondly, notice that by Assumption~\ref{asm:g_bound}, because any function in $\cG_n$ is upper bounded by $M_u$ and lower bounded by $1/M_l$, the bracketing pairs $\br{g^L,g^U}$ are also bounded functions; otherwise we can modify the pair to be clipped version: $\tilde g^L(x) = \min\br{\max\br{g^L(x),1/M_l},M_u}$, $\tilde g^U(x) = \min\br{\max\br{g^U(x),1/M_l},M_u}$. Now we can show that $\absbr{\sqrt{g^U} - \sqrt{q^L}} = \frac{\absbr{g^U-g^L}}{\sqrt{g^U} + \sqrt{g^L}}\leq 2\sqrt{M_l}\absbr{g^U-g^L}$, and $\nbr{\sqrt{g^U} - \sqrt{g^L}}_{L^2(\bQ)}\leq 2\sqrt{M_l} \nbr{g^U-g^L}_{L^2(\bQ)}\leq 2\sqrt{M_l}\delta$. Hence $$N_{[]}\sbr{2\sqrt{M_l}\delta,\cF_{\delta,n}^Q,L^2(\bQ)}\leq N_{[]}\sbr{\delta,\cG_n,L^2(Q)},$$ and note that 
\begin{align*}
    \tilde J_{[]}\sbr{\delta, \cF_{\delta,n}^\bQ, L^2(\bQ)} &= \int_0^{\delta}\sqrt{1+\log N_{[]}(\epsilon,\cF_{\delta,n}^\bQ,L^2(\bQ))}\dd\epsilon \leq \int_0^{\delta}\sqrt{1+\log N_{[]}\sbr{\frac{\epsilon}{2\sqrt{M_l}},\cG_n,L^2(\bQ)}}\dd\epsilon\\
    &\overset{t = \frac{\epsilon}{2\sqrt{M_l}}}{=} 2\sqrt{M_l}\int_{0}^{\frac{\delta}{2\sqrt{M_l}}} \sqrt{1+\log N_{[]}\sbr{t,\cG_n,L^2(\bQ)}}\dd t \\
    &= 2\sqrt{M_l}\tilde J_{[]}\sbr{\frac{\delta}{2\sqrt{M_l}}, \cG_n,L^2(Q)},
\end{align*}
hence
\begin{equation}
    \bE \|\bG_n^Q\|_{\cF_{\delta,n}^Q} \leq 2\sqrt{M_l}\tilde J_{[]}\sbr{\delta/(2\sqrt{M_l}), \cG_n, L^2(\bQ)} \sbr{1+\frac{2\sqrt{M_l}\tilde J_{[]}\sbr{\delta/(2\sqrt{M_l}), \cG_n, L^2(\bQ)}}{\delta^2\sqrt{n}}M}. \label{eq:Q_phi_n}
\end{equation}

For the upper bound of $\bE \|\bG_n^P\|_{\cF_{\delta,n}^P}$, we can similarly apply Lemma 3.4.5 in \citet{van1996weak}. To verify that $P\br{\sbr{g^{-1/2}-g_n^{-1/2}}^2}\lesssim \delta^2$ given that in $\cF_{\delta,n}^P$, $h_Q(g_n,g)<\delta$, we refer to the previous proof in \eqref{eq:pf2_verify} and replace $g$ by $g_n$ in \eqref{eq:pf2_verify}. Similarly, $\nbr{\frac{1}{\sqrt{g^L}} - \frac{1}{\sqrt{g^U}}}_{L^2(\bP)} = \nbr{ \frac{g^L-g^U}{\sqrt{g^Ug^L}\sbr{\sqrt{g^U} + \sqrt{g^L}}} }_{L^2(\bP)}\leq \frac{1}{2M_l^{3/2}}\nbr{g^L-g^U}_{L^2(\bP)}\leq \frac{\delta}{2M_l^{3/2}}$. Hence
\begin{align}
    \bE \|\bG_n^P\|_{\cF_{\delta,n}^P} 
        &\leq \frac{1}{2M_l^{3/2}}\tilde J_{[]}\sbr{2M_l^{3/2}\delta, \cG_n, L^2(\bP)} \sbr{1+\frac{\frac{1}{2M_l^{3/2}}\tilde J_{[]}\sbr{2M_l^{3/2}\delta, \cG_n, L^2(\bP)}}{\delta^2\sqrt{n}}M}.
        \label{eq:phi_n}
\end{align}

For simplicity, we unify the upper bounds in \eqref{eq:Q_phi_n} and \eqref{eq:phi_n}. For the upper bound in \eqref{eq:Q_phi_n}, note $\|f\|^2_{L^2(\bP)}\geq \frac{1}{M_l}\|f\|^2_{L^2(\bQ)}$, $N_{[]}\sbr{\epsilon,\cG_n,L^2(\bQ)}\leq N_{[]}\sbr{\epsilon/\sqrt{M_l},\cG_n,L^2(\bP)}$, and 
\begin{align*}
    \tilde J_{[]}\sbr{\delta/(2\sqrt{M_l}), \cG_n, L^2(\bQ)} &= \int_0^{\delta/(2\sqrt{M_l})}\sqrt{1+N_{[]}\sbr{\epsilon,\cG_n,L^2(\bQ)}}\dd\epsilon \\
    &\leq \int_0^{\delta/(2\sqrt{M_l})} \sqrt{1+N_{[]}\sbr{\epsilon/\sqrt{M_l},\cG_n,L^2(\bP)}}\dd\epsilon
    =\sqrt{M_l} \tilde J_{[]}\sbr{\frac{\delta}{2M_l},\cG_n,L^2(\bP)}.
\end{align*}
Recall that in Assumption~\ref{asm:g_bound}, $\sup_{g\in\cG_n}|g^{-1}|\leq M_l$, WLOG, let $2M_l>1$. Hence $$\tilde J_{[]}\sbr{\delta/(2\sqrt{M_l}), \cG_n, L^2(\bQ)}\lesssim \tilde J_{[]}\sbr{\delta,\cG_n,L^2(\bP)}.$$

 For the upper bound in \eqref{eq:phi_n}, WLOG, let $2M_l^{3/2}>1$, $N_{[]}\sbr{2M_l^{3/2}\epsilon,\cG_n,L^2(\bP)}\leq N_{[]}\sbr{\epsilon,\cG_n,L^2(\bP)}$. Hence 
 \begin{align*}
     \frac{1}{2M_l^{3/2}}\tilde J_{[]}\sbr{2M_l^{3/2}\delta, \cG_n, L^2(\bP)} &= \frac{1}{2M_l^{3/2}}\int_0^{2M_l^{3/2}\delta} \sqrt{1+N_{[]}\sbr{\epsilon,\cG_n,L^2(\bP)}}\dd\epsilon \\
     &\overset{t=\frac{\epsilon}{2M_l^{3/2}}}{=}\int_0^\delta\sqrt{1+N_{[]}\sbr{2M_l^{3/2}t,\cG_n,L^2(\bP)}}\dd t\\
     &\leq \tilde J_{[]}\sbr{\delta,\cG_n,L^2(\bP)}.
 \end{align*}

Now we have shown that both the RHS of \eqref{eq:Q_phi_n} and \eqref{eq:phi_n} are upper bounded by $$ \tilde J_{[]}\sbr{\delta, \cG_n, L^2(\bP)} \sbr{1+\frac{\tilde J_{[]}\sbr{\delta, \cG_n, L^2(\bP)}}{\delta^2\sqrt{n}}M}$$ up to a constant that only depends on $M_u,M_l$.
Let $\phi_n(\delta) = \tilde J_{[]}\sbr{\delta, \cG_n, L^2(\bP)} \sbr{1+\frac{\tilde J_{[]}\sbr{\delta, \cG_n, L^2(\bP)}}{\delta^2\sqrt{n}}M}$.  To solve for $r_n$ based on Conditions 2 and 3, we need to check that $\phi_n(\delta)/\delta^\alpha$ is decreasing on $(\delta_n,\infty)$ for some $0<\alpha<2$.

Note that $\tilde J_{[]}\sbr{\delta,\cG_n, L^2(\bP)} =  \int_{0}^\delta \sqrt{1+\log N_{[]}\sbr{\epsilon,\cG_n, L^2(\bP)}}\dd\epsilon$, and the integrand $$f(\epsilon)=\sqrt{1+\log N_{[]}\sbr{\epsilon,\cG_n, L^2(\bP)}}$$ is a decreasing function of $\epsilon$, the radius of the bracketing number.  Let $h(\delta) = \frac{\int_{0}^\delta f(\epsilon)\dd\epsilon}{\delta^\alpha}$, and let $\alpha=1$
\begin{align*}
    h'(\delta)&=  \frac{f(\delta)\delta -\int_{0}^\delta f(\epsilon)\dd\epsilon }{\delta^{2}} = \frac{\int_0^\delta (f(\delta) - f(\epsilon))\dd\epsilon}{\delta^{2}} \leq 0.
\end{align*}
Hence $h(\delta)$ is a decreasing function $\delta$, and $\phi_n(\delta)/\delta = h(\delta)\sbr{1+\frac{h(\delta)}{\delta\sqrt{n}}M}$ is also a decreasing function of $\delta$.

 Now we've found the $\phi_n(\delta)$ that satisfies Condition 2, and it depends on the neural network parameter bounds, which will be determined later when solving for $r_n$ in Condition 3.
In Condition 3, the inequality $r_n^2\phi(1/r_n)\lesssim \sqrt{n}$ is equivalent to $r_n^2\tilde J_{[]}\sbr{1/r_n, \cG_n, L^2(\bP)} \lesssim\sqrt{n}$ (To see this, if we denote $x_n = r_n^2\tilde J_{[]}\sbr{1/r_n, \cG_n, L^2(\bP)}$, then $r_n^2\phi(1/r_n)=x_n(1+\frac{x_n}{\sqrt{n}}M)$, and note that $x_n+\frac{x_n^2}{\sqrt{n}}M\leq 2\sqrt{n} \Leftrightarrow C_Mx_n\leq\sqrt{n}$ by solving the quadratic roots, where $C_M=\frac{\sqrt{1+8M}-1}{2M}$ when $x_n>0$).

Note that $\tilde J_{[]}\sbr{\delta,\cG_n, L^2(\bP)} =  \int_{0}^\delta \sqrt{1+\log N_{[]}\sbr{\epsilon,\cG_n, L^2(\bP)}}\dd\epsilon$. We can further simplify the result in Lemma \ref{lem:nn_bracket_number} and let $A_n = S_n\log\sbr{B_n^{L_n}U_n^{L_n}L_n(2U_nL_n+d)e\|H\|_{L^2(\bP)}}-S_n\log\sbr{2S_n}+1$, then the above integrand can be upper bounded by 
\[\sqrt{1+\log N_{[]}\sbr{\epsilon,\cG_n, L^2(\bP)}}\leq \sqrt{A_n - S_n\log\epsilon}.\]
We can further bound the integral $\int_0^\delta \sqrt{A_n-S_n\log\epsilon}\dd\epsilon$ by the Cauchy–Schwarz inequality, 
\[\int_0^\delta \sqrt{A_n-S_n\log\epsilon}\dd\epsilon \leq \sqrt{\delta} \sqrt{\int_0^\delta A_n - S_n\log\epsilon}\dd\epsilon =\delta \sqrt{A_n-S_n\log\delta + S_n}. \]

Condition 4 is true because both $\hat g_n$ and $g_n$ belong to $\cG_n$, and $\hat g_n$ is the maximizer of $\bM_n(g)$ over $\cG_n$.

Note that $\delta_n \asymp h_\bQ(g_n,g_0)$, and $h_\bQ(g_n,g_0)$ depends on the distance between the sieved neural network space where $g_n$ belongs to and the anisotropic Besov space where $g_0$ belongs to. To get $h_\bQ(g_n,g_0)$, we use Proposition 2 in \cite{suzuki2021deep}, where the neural network parameters used in Lemma~\ref{lem:nn_bracket_number} are specified in the following way,
\begin{align*}
L_N &:= 3 + 2 \Big\lceil \log_{2}\!\left( \frac{3^{d \vee m}}{\epsilon c_{(d,m)}} \right) + 5 \Big\rceil 
            \Big\lceil \log_{2}(d \vee m) \Big\rceil, 
 \quad U_N := N U_{0},  \\
S_N &:= \big[(L_N-1)U_{0}^{2} + 1\big] N, 
 \quad B_N := O\!\left(N^{d(1+\nu^{-1})(1/s-\tilde{\beta})_{+}}\right), \numberthis\label{eq:nn_params}
\end{align*}
where $\epsilon = N^{-\tilde{\beta}} \log(N)^{-1}$, $m$ is a positive integer such that $0<\bar\beta<\min(m,m-1+1/s)$, $\delta = (1/s - 1/r)_{+}$ (we use the $L^2(\cX)$ norm in \eqref{eq:minimax_distance} and let $r=2$), $\nu = (\tilde{\beta} - \delta)/(2\delta)$, and $ U_0 := 6dm(m+2) + 2d$. $c_{(d.m)}$ is a constant that only depends on $d$ and $m$. The parameter $N$ controls the total number of neural network parameters, and we will allow $N$ to grow with the sample size $N=N(n)$ as $n\to\infty$. Hence we abbreviate the notations $\br{L_N,U_N,S_N,B_N}$ to $\br{L_n,U_n,S_n,B_n}$ in the following discussion. Proposition 2 in \cite{suzuki2021deep} states that, under \eqref{eq:nn_params}, 
\begin{align}
    \sup_{g_0\in\cU\sbr{B^{\beta}_{s,t}(\cX)}}\inf_{g_n\in \cG_n} \|g_n-g_0\|_{L^2(\cX)} \lesssim N^{-\tilde\beta}. \numberthis \label{eq:minimax_distance}
\end{align}
Recall Assumption~\ref{asm:g_bound}, $\|g_0^{-1}\|_\infty\leq M_l$. With the assumption $\|q_0^2/p_0\|_\infty <\infty$, we can derive the following inequality between distances:
\begin{align*}
    h^2_\bQ(g_n,g_0) &= \int\sbr{ g_n^{1/2}-g_0^{1/2}}^2 \dd\bQ = \int\sbr{ \frac{g_n-g_0}{\sqrt{g_n}+\sqrt{g_0}}}^2 \dd\bQ \leq \int\sbr{g_n-g_0}^2 \frac{q_0}{g_0}\dd \leb(x) \\
    &\leq \|q_0^2/p_0\|_\infty  \|g_n-g_0\|^2_{L^2(\cX)} \lesssim \|g_n-g_0\|^2_{L^2(\cX)} \overset{(*)}{\lesssim} N^{-2\tilde\beta}.
\end{align*}
Note that $(*)$ is based on \eqref{eq:minimax_distance} and $g_n$ is taken to be the projection of $g_0$ in $\cG_n$.

Hence $\delta_n \lesssim N^{-\tilde\beta}$. Now we need to solve for an upper bound for $r_n$ such that
\begin{itemize}
    \item[(i)] $r_n\sqrt{A_n+S_n+S_n\log(r_n)}\lesssim \sqrt{n}$;
    \item[(ii)] $\frac{1}{r_n} \gtrsim  N^{-\tilde\beta}$.
\end{itemize}

We first choose $N \asymp r_n^{1/\tilde\beta}$ based on (ii). The inequality (i) depends on the neural network parameters. Recall \eqref{eq:nn_params}, we can compute the order of neural network configuration parameters w.r.t. $N$ first: $L_n = O(\log(N)), U_n = O(N), S_n=O(N\log(N)), B_n = O(N^\eta)$, where $\eta$ is a constant that depends on the smoothness $\tilde\beta$ and the input dimension $d$. Here, we consider $d$ as a constant that does not grow with $n$.

We can simplify the first inequality (i) in the order of $N$ as follows
\begin{align*}
    A_n+S_n &= S_n\log\sbr{B_n^{L_n}U_n^{L_n}L_n(2U_nL_n+d)e\|H\|_{L^2(\bP)}}-S_n\log\sbr{2S_n}+S_n+1\\
    &=S_n\br{L_n\log(B_n) + L_n\log(U_n) + \log(2U_nL_n+d)+ \log(L_n)+O(1) } - S_n\log(2S_n)+S_n+1\\
    &= O(N\sbr{\log N}^3).
\end{align*}
Because we set $N \asymp r_n^{1/\tilde\beta}$, $r_n\sqrt{A_n+S_n+S_n\log(r_n)} \asymp r_n^{1+\frac{1}{2\tilde\beta}}\sbr{\log(r_n)}^{3/2}$. Now we only need to solve for an upper bound of $r_n$ from 
\begin{equation}\label{eq:r_upperbound}
 r_n^{1+\frac{1}{2\tilde\beta}}\sbr{\log(r_n)}^{3/2} \lesssim \sqrt{n}.
\end{equation}

To solve for this inequality, we can guess that the upper bound of $r_n$ takes the form $r_n = \frac{n^{c_1}}{(\log n)^{c_2}}$, because the RHS of \eqref{eq:r_upperbound} is of polynomial rate of $n$, whereas the LHS of \eqref{eq:r_upperbound} is a polynomial of $r_n$ times a log rate of $r_n$. Then we solve for $c_1,c_2$ by first matching the polynomial rate, $c_1(1+\frac{1}{2\tilde\beta}) = 1/2$, $c_1 = \frac{\tilde\beta}{2\tilde\beta+1}$. The log term in the LHS is $(\log(r_n))^{3/2} = O\sbr{\sbr{\frac{\tilde\beta}{2\tilde\beta+1}\log n - c_2\log\log n}^{3/2}} = O\sbr{\sbr{\log n}^{3/2}}$. Now the rate on the LHS is 
\begin{align*}
    r_n^{1+\frac{1}{2\tilde\tilde\beta}}\sbr{\log(r_n)}^{3/2} &= O\sbr{n^{1/2} (\log n)^{-c_2(1+\frac{1}{2\tilde\beta})} (\log n)^{3/2}}.
\end{align*}
Hence by setting $c_2(1+\frac{1}{2\tilde\beta}) = 3/2$, we can solve for $c_2 = \frac{3\tilde\beta}{2\tilde\beta+1}$. Hence the upper bound for $r_n$ is $O(n^{\frac{\tilde\beta}{2\tilde\beta+1}}\sbr{\log n}^{-\frac{3\tilde\beta}{2\tilde\beta+1}})$.

Now the proof for the estimation error is complete, 
\[h_\bQ(\hat g_n,g_n) =O_p(1/r_n)= O_p(n^{-\frac{\tilde\beta}{2\tilde\beta+1}}\sbr{\log n}^{\frac{3\tilde\beta}{2\tilde\beta+1}}).\]
The approximation error $h_\bQ(g_n,g_0) \leq O(N^{-\tilde\beta})$ follows the same rate by design $r_n=O(N^{\tilde\beta})$. Hence the convergence rate is $O_p(n^{-\frac{\tilde\beta}{2\tilde\beta+1}}\sbr{\log n}^{\frac{3\tilde\beta}{2\tilde\beta+1}})$.

\end{proof}

\newpage
\section{Output activation function and sensitivity analysis}\label{supp_sec:act_sensi}

As discussed in Section~\ref{subsec:connection}, the different choices of the last-layer output activation function may impact the estimation results. Our main analysis uses the bounded sigmoid function $\sigma^{\text{sig}}_\alpha(x) = \frac{1}{1-\exp\br{-\alpha x}}$ where $\alpha$ controls how fast the activation function saturates to the two end points $0,2$. In an earlier version of this work, we have experimented with the bounded softplus function, $\sigma^{\text{bsp}}(x) = \frac{2f(x)}{f(x)+1}$ where $f(x) = \log\sbr{1+e^{x}}$. Figure~\ref{fig:act_compare} shows these different activation functions. The sigmoid function $\sigma^{\text{sig}}_\alpha(x)$ saturates to the two endpoints $0,2$ faster as $\alpha$ grows, whereas the bounded softplus function $\sigma^{\text{bsp}}(x)$ saturates asymmetrically to 0 and to 2. In the main analysis, we use the sigmoid function with $\alpha=1$. In this section, we include a sensitivity analysis on the choice of different output activation functions.

\begin{figure}
    \centering
    \includegraphics[width=0.7\linewidth]{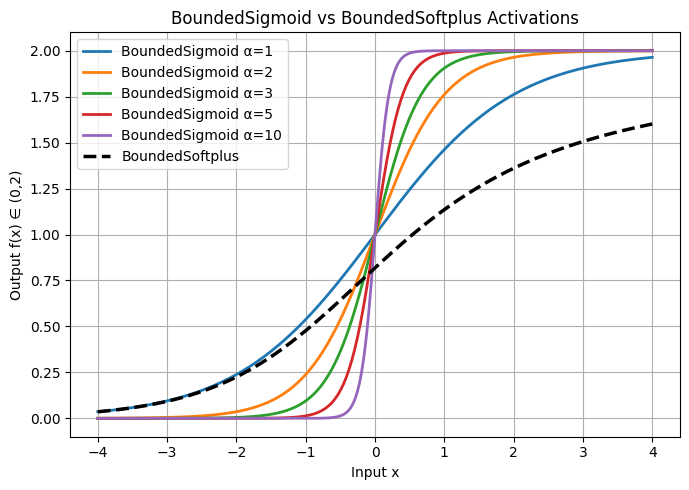}
    \caption{Comparison of different output activation functions.}
    \label{fig:act_compare}
\end{figure}

\noindent\textbf{1D Gaussian.} For the 1D Gaussian example in Figure~\ref{fig:oneD_compare}, Figure~\ref{fig:oneD_compare_softplus} shows the estimation result if we use the bounded softplus output activation $\sigma^{\text{bsp}}(x)$. We can see that the asymmetric saturation of $\sigma^{\text{bsp}}(x)$ leads to under-estimated $\hat r(x)$ near 2.
\begin{figure}
    \centering
    \includegraphics[width=0.9\linewidth]{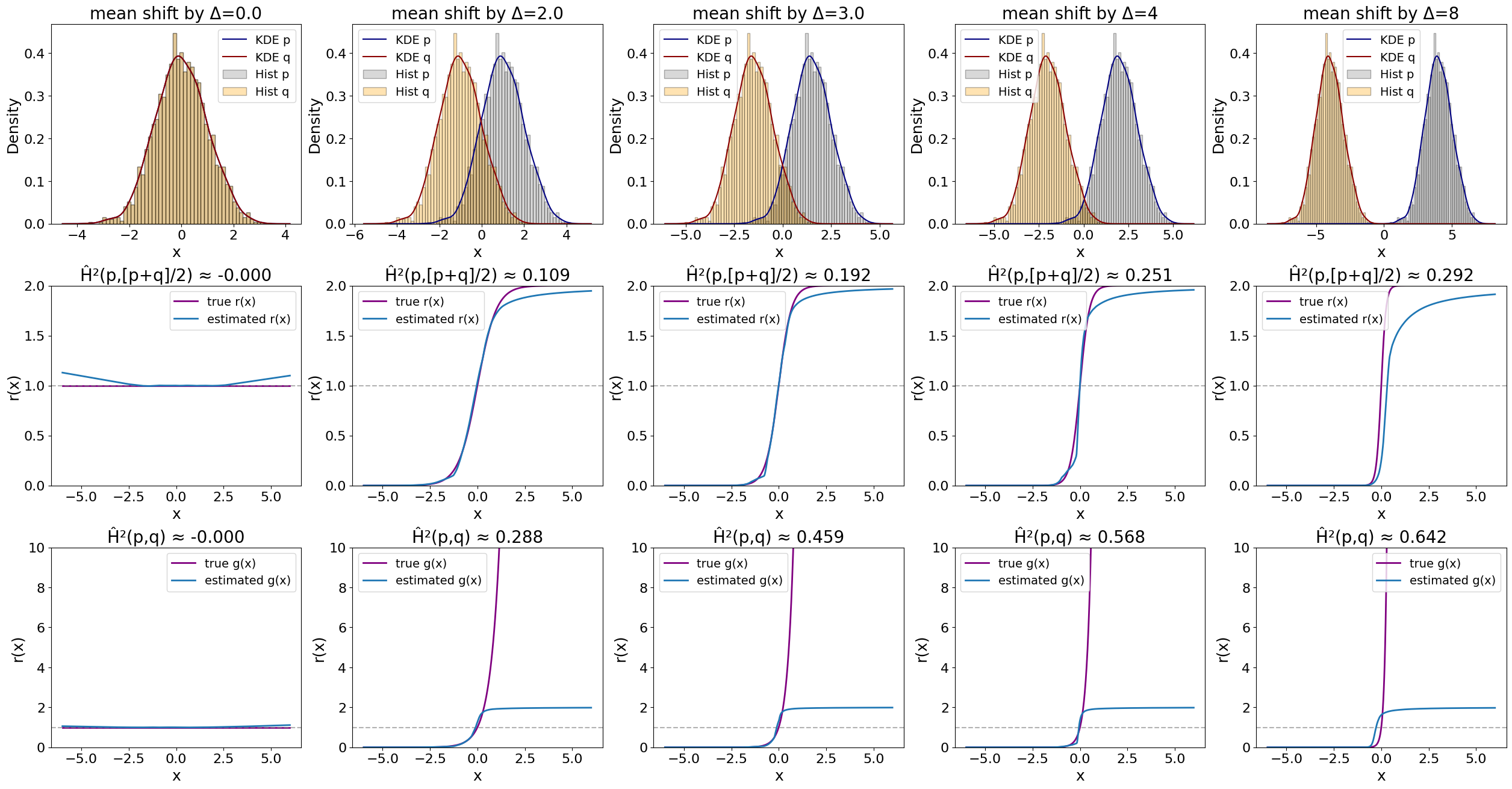}
    \caption{1D Gaussian example using the bounded softplus output activation $\sigma^{\text{bsp}}(x)$.}
    \label{fig:oneD_compare_softplus}
\end{figure}

\noindent\textbf{MNIST.} For the main result in Section~\ref{subsec:mnist}, because VAE fails to generate high-quality images, the estimated $\hat r(X)$ tends to have spikes of mass near 0 and 2, and such results do not change much regardless of the different output activation. Here, we focus on the sensitivity analysis of the DCGAN generator. Figure~\ref{fig:MNIST_DCGAN_softplus} presents the results using bounded softplus output activation, which has an asymmetric saturation to 0 and 2. Figure~\ref{fig:MNIST_DCGAN_sigmoid2} presents the results using the bounded sigmoid output activation function with a faster saturation rate $\alpha=2$. The asymmetric saturation of $\sigma^{\text{bsp}}(x)$ leads to the potential underestimated $\hat r(X)$ near 2 as shown in Figure~\ref{fig:MNIST_DCGAN_softplus}, a consistent pattern as in Figure~\ref{fig:oneD_compare_softplus}. The faster saturation brought by $\sigma^{\text{sig}}_{\alpha=2}(x)$ in Figure~\ref{fig:MNIST_DCGAN_sigmoid2} leads to a more concentrated mass, especially around 0, compared to Figure~\ref{fig:MNIST_DCGAN}. However, the overall pattern for the DCGAN performance remains the same: the images ``1" are under-represented by DCGAN due to its geometric simplicity, and images with more round-shaped patterns like ``5" and ``8" are more commonly drawn by DCGAN.

\begin{figure}
    \centering
    \includegraphics[width=0.8\linewidth]{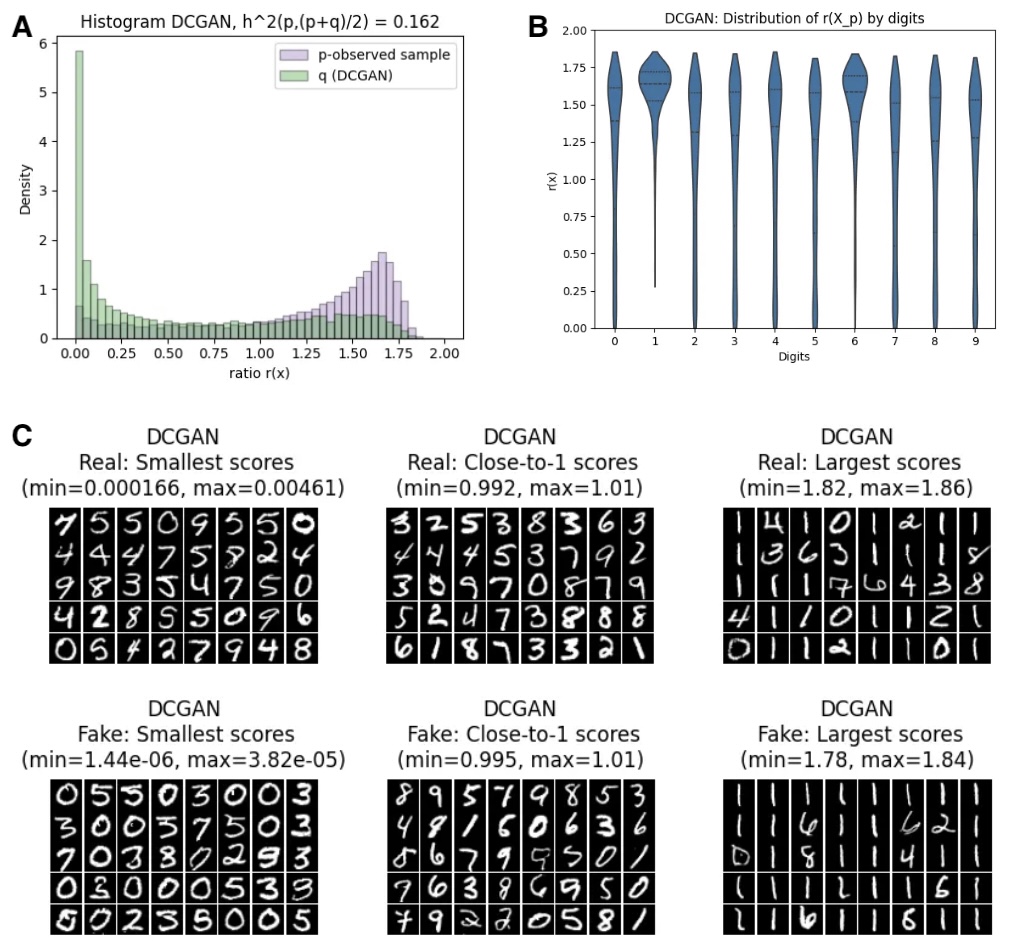}
    \caption{MNIST DCGAN results when using the bounded softplus output activation function $\sigma^{\text{bsp}}(x)$.}
    \label{fig:MNIST_DCGAN_softplus}
\end{figure}

\begin{figure}
    \centering
    \includegraphics[width=0.8\linewidth]{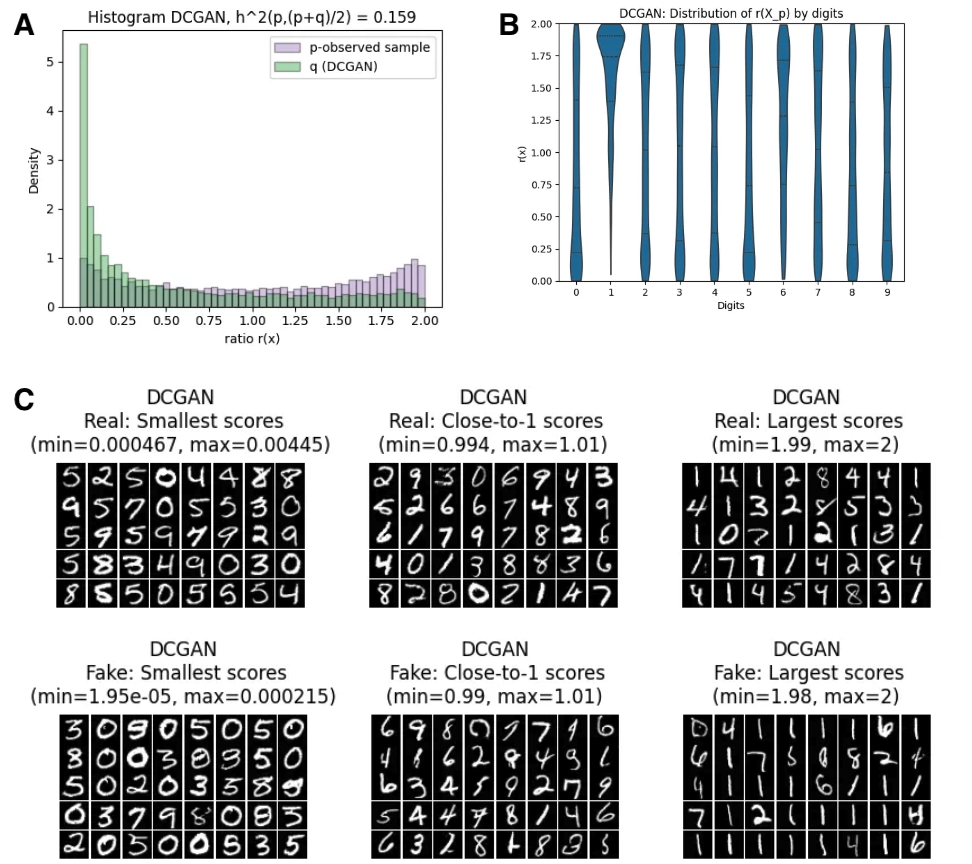}
    \caption{MNIST DCGAN results when using the bounded sigmoid output activation function with a faster saturation rate $\alpha=2$, $\sigma^{\text{sig}}_{\alpha=2}(x)$.}
    \label{fig:MNIST_DCGAN_sigmoid2}
\end{figure}

\begin{figure}[ht!]
    \centering
    \begin{subfigure}{0.3\textwidth}
        \includegraphics[width=\linewidth]{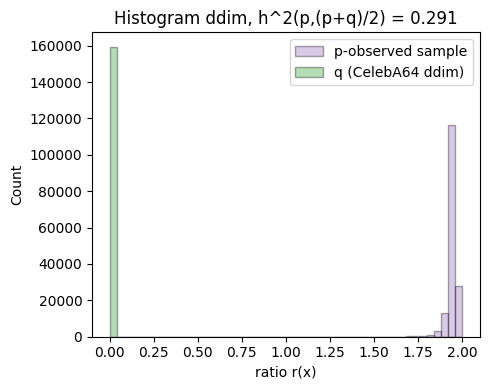}
        \caption{Histogram.}
        \label{fig:hist_ddim}
    \end{subfigure}%
    \hfill
    \begin{subfigure}{0.27\textwidth}
        \centering
        \resizebox{\linewidth}{!}{%
        \begin{tabular}{lll}
        \toprule
       & $r(X_p)$ & $r(X_q)$ \\
       \hline
        length & 162770  & 160000  \\
        mean   & 1.94       & 0.00       \\
        std    & 0.04       & 0.06       \\
        min    & 0.00       & 0.00       \\
        q1     & 1.94       & 0.00       \\
        median & 1.95       & 0.00       \\
        q3     & 1.96       & 0.00       \\
        max    & 1.98       & 1.95      \\
        \bottomrule
        \end{tabular}
        }
        \caption{Summary Statistics.}
        \label{tb:summary_ddim}
    \end{subfigure}
    \hfill
    \begin{subfigure}{0.27\textwidth}
        \centering
        \resizebox{\linewidth}{!}{%
        \begin{tabular}{lll}
        \toprule
                      & coef  & pvalue \\
                      \hline
            Wearing\_Hat          & -1.257 & 0.000  \\
            No\_Beard             & 1.010  & 0.002  \\
            Young                 & 0.831  & 0.003  \\
            5\_o\_Clock\_Shadow   & 2.213  & 0.003  \\
            Mouth\_Slightly\_Open & -0.792 & 0.004  \\
            Oval\_Face            & 1.974  & 0.007  \\
            Goatee                & 1.382  & 0.015  \\
            Bangs                 & -0.821 & 0.041  \\
            Bags\_Under\_Eyes     & 0.834  & 0.054  \\
            Wavy\_Hair            & 0.870  & 0.078 \\
            \bottomrule
        \end{tabular}
        }
        \caption{Logistic regression on $P(r(X_p)>1)$.}
        \label{fig:ddim_logistic_regression}
    \end{subfigure}

    \vskip\baselineskip
    \begin{subfigure}{0.98\textwidth}
        \includegraphics[width=\linewidth]{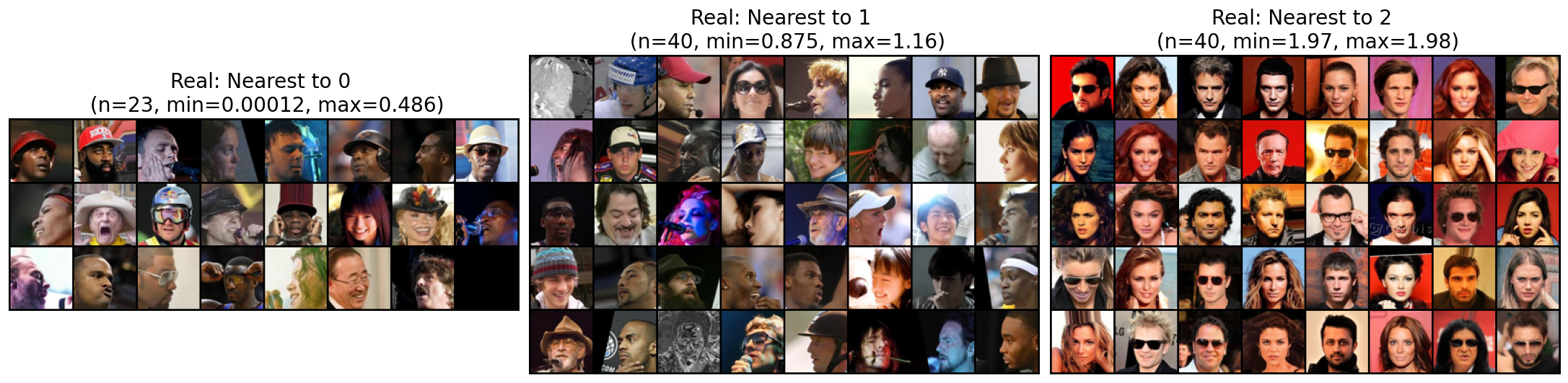}
        \caption{Real images.}
        \label{fig:CelebA_DDIM_real_img}
    \end{subfigure}

    \vskip\baselineskip
    \begin{subfigure}{0.98\textwidth}
        \includegraphics[width=\linewidth]{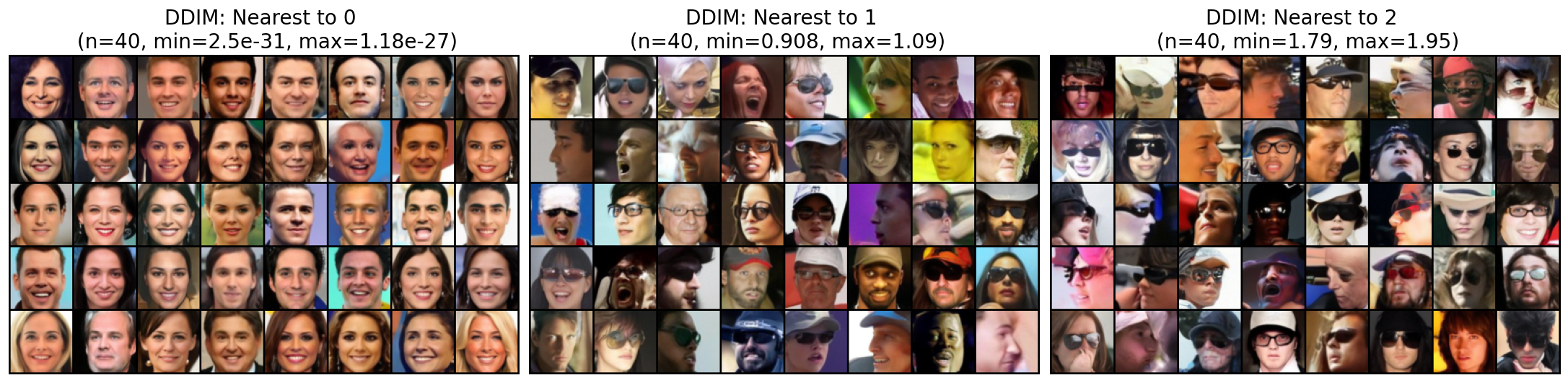}
        \caption{Fake images.}
        \label{fig:CelebA_DDIM_fake_img}
    \end{subfigure}

    \caption{CelebA-64: DDIM evaluated on full training data using bounded softplus as the output activation function.}
    \label{fig:CelebA_DDIM_asym}
\end{figure}

\noindent\textbf{CelebA.} The main results in Section~\ref{subsec:celeb} are based on the bounded sigmoid function with $\alpha=0.1$. Figure~\ref{fig:CelebA_DDIM_asym} presents the results when using the bounded softplus function, an asymmetric function, as the output activation. Additionally, we also included the logistic regression in Table~\ref{fig:ddim_logistic_regression} for the downstream feature detection. Although the numerical values of $r(X_p),r(X_q)$ differ from the main analysis in Section~\ref{subsec:celeb}, the general trend is the same: high concentration of homogeneous patterns of face images (left panel in Figure~\ref{fig:CelebA_DDIM_fake_img}), with $r(X_q)$ concentrating around 0, shows possible model hallucinations; the most easily learned feature by DDIM is still the \textit{wearing\_hat feature,} and the most difficult one is the \textit{5\_o\_Clock\_Shadow}. The major difference from Section~\ref{subsec:celeb} is that the evaluated $r(X_p),r(X_q)$ cannot reach the upper bound 2, due to the asymmetric activation function.

\section{Additional Numerical Experiments}

This section provides additional results for the numerical experiments in Section~\ref{sec:numerical_example}. 

For the CelebA-64 data, Figure~\ref{fig:DDIM_900_outsupport} shows 900 images that have small $r(X_q)$ values around 0 and display a homogeneous pattern. 

For the American Gut Project microbiome data, Figure~\ref{fig:agp_icfm_alpha_diversity} shows the alpha diversity plots in terms of the Shannon diversity and Gini-Simpson diversity. Figure~\ref{fig:icfm_sunburst} provides the sunburst plot on each level of the taxonomy table, where we conduct correlation analysis between the CLR-transformed relative abundance at each taxonomic level with the $r(X)$ evaluated in both observed and synthetic data. Positive correlation indicates the corresponding type of taxa is more abundant in the observed data, and negative correlation indicates more abundance in the synthetic data. Table~\ref{tab:class_association} provides the numeric value of the correlation coefficients on the class level, with the top 10 classes of microbiome ranked by the absolute value of the correlation coefficient.

\begin{figure}[ht!]   
    \centering
    \includegraphics[width=0.9\linewidth]{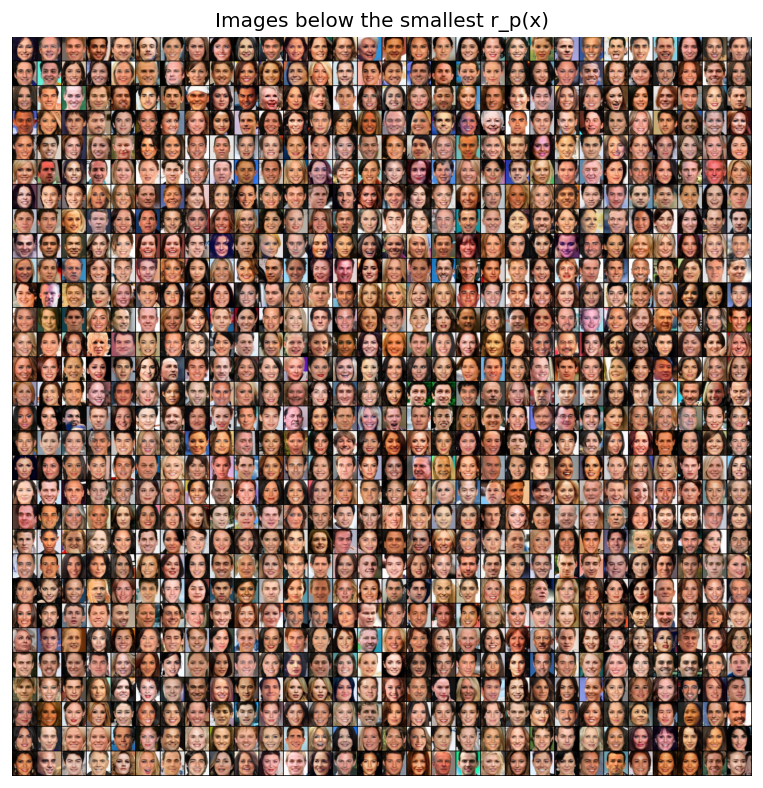}
    \caption{900 generated images by DDIM with $r(X_q)$ smaller than the leaset $r(X_p)$ value. There are a total of 159719 out of 160000 such samples.}
    \label{fig:DDIM_900_outsupport}
\end{figure}

\begin{figure}[ht!]   
    \centering
    \includegraphics[width=0.5\linewidth]{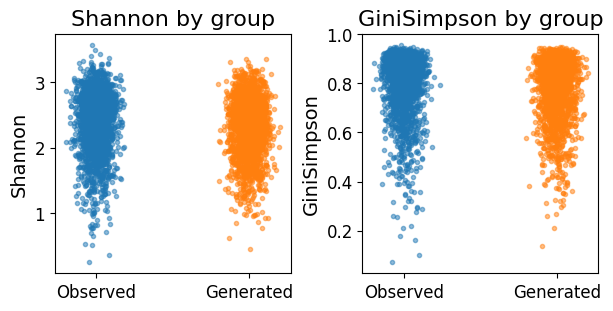}
    \caption{ICFM alpha diversity (within sample variations): (1) Shannon diversity: $H(x) = - \sum_{j=1}^{S} x_j \log x_j$; (2) Gini-Simpson diversity: $D_{\text{GS}}(x) = 1 - \sum_{j=1}^{S} x_j^2$. }
    \label{fig:agp_icfm_alpha_diversity}
\end{figure}
\begin{figure}
    \centering
    \includegraphics[width=0.9\linewidth]{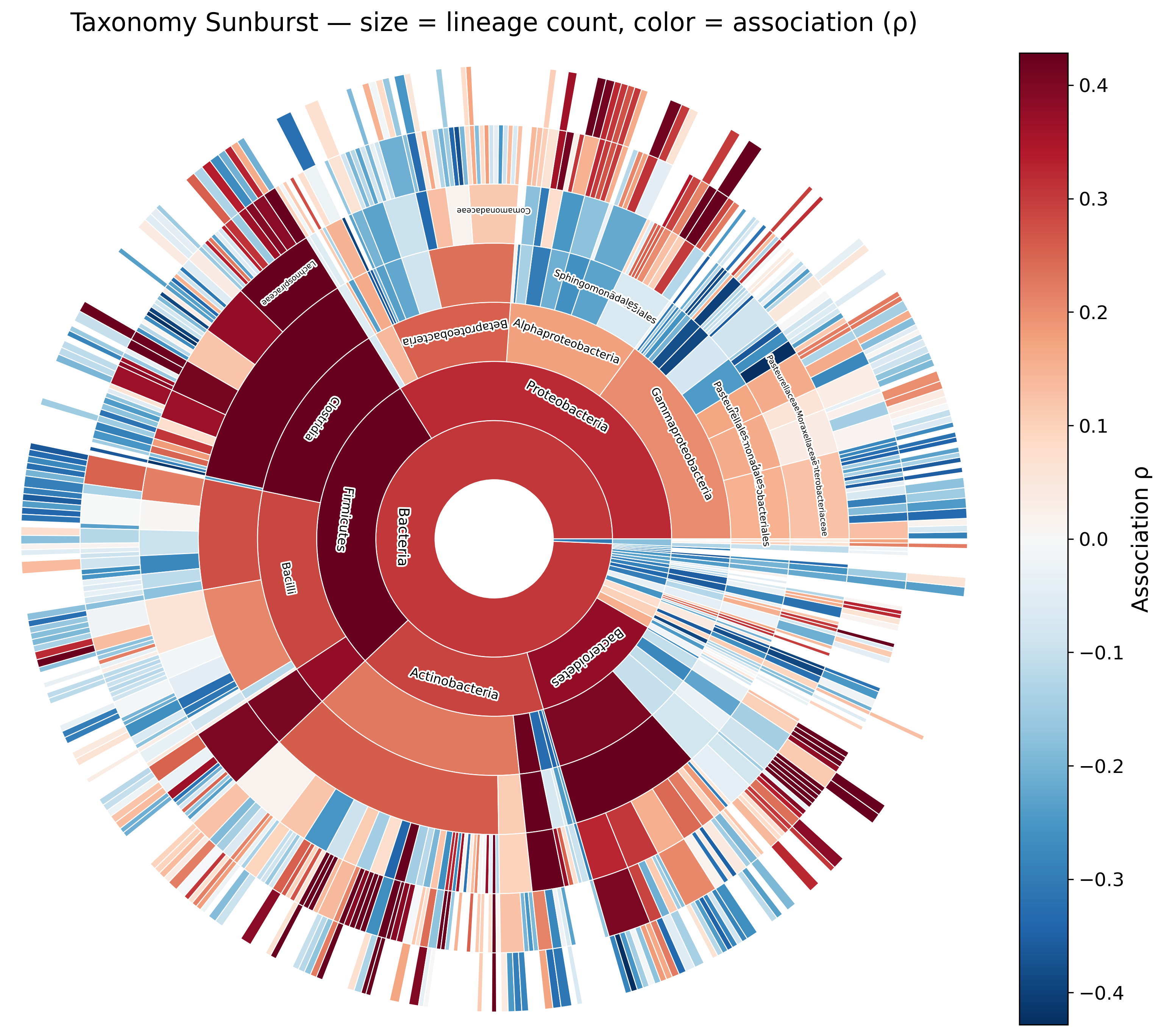}
    \caption{Association test on each level of the taxonomy. Positive (red) values are the taxa whose relative abundance is associated with higher $r(X)$ values, and are more common (relative to the community geometric mean of that level) in the observed data than in the generated data.}
    \label{fig:icfm_sunburst}
\end{figure}

\begin{table}[ht!]   
    \centering
    \begin{tabular}{ll}
    \toprule
   Class                   & rho    \\
Clostridia          & 0.477  \\
Coriobacteriia      & 0.421  \\
Bacteroidia         & 0.404  \\
Opitutae            & -0.389 \\
Erysipelotrichi     & 0.376  \\
Elusimicrobia       & -0.363 \\
Spartobacteria     & -0.357 \\
Synergistia         & -0.357 \\
Deinococci          & -0.354 \\
Acidimicrobiia      & -0.340\\
\bottomrule
\end{tabular}
    \caption{Class-level association test, between $r(X)$ and the CLR-transformed class-level relative abundance data.}
    \label{tab:class_association}
\end{table}

\section{Simulation with Known Densities}\label{sec:sim}

Although the goal of this work is not to propose a new density ratio estimation method, it is still necessary to check the point estimation accuracy of the neural network estimator on known densities and compare its performance with existing density ratio estimation methods. In this simulation study, we start with a 2D case and extend to four high-dimensional cases with a latent 2D manifold.

In the 2D case, we simulate $P$ and $Q$ from mixtures of 3 normal densities, as shown in the first plot in Figure~\ref{fig:sim}. The two distributions \(P\) and \(Q\) are defined as three-component Gaussian mixtures
with shared mixture weights \(\boldsymbol{w} = (w_1, w_2, w_3) = (0.3, 0.3, 0.4)\):
\begin{align}
    P(x)
    &= \sum_{k=1}^{3} w_k \, 
    \mathcal{N}\!\left(x \mid \boldsymbol{\mu}_k^{(P)}, \, \Sigma_k^{(P)} \right), \label{eq:P_mixture}\\[4pt]
    Q(x)
    &= \sum_{k=1}^{3} w_k \, 
    \mathcal{N}\!\left(x \mid \boldsymbol{\mu}_k^{(Q)}, \, \Sigma_k^{(Q)} \right). \label{eq:Q_mixture}
\end{align}
The 3-component mixture of normal densities represents 3 cases of distributional differences as shown in the first scatterplot in Figure~\ref{fig:sim}: (1) the top-left group is the component where $P$ and $Q$ have the same density; (2) the top-right group is when $P$ and $Q$ have the same mean, but the covariance is different; (3) the bottom-left group is when $P$ and $Q$ have the same covariance structure, but the means are shifted. The numerical values of the normal parameters are set up as follows:
\begin{enumerate}
  \item[\textbf{(1)}] \textbf{Top-left group (identical component).}  
  Both \(P\) and \(Q\) share the same mean and covariance:
  \[
  \boldsymbol{\mu}_{1}^{(P)} = \boldsymbol{\mu}_{1}^{(Q)} = (-1,\,5)^{\!\top}, \qquad
  \Sigma_{1}^{(P)} = \Sigma_{1}^{(Q)} =
  \begin{pmatrix}
  1 & 0 \\[2pt]
  0 & 1
  \end{pmatrix}.
  \]

  \item[\textbf{(2)}] \textbf{Top-right group (different covariance, same mean).}  
  Both distributions share the same mean but have different covariance structures:
  \[
  \boldsymbol{\mu}_{2}^{(P)} = \boldsymbol{\mu}_{2}^{(Q)} = (5,\,5)^{\!\top}, \qquad
  \Sigma_{2}^{(P)} =
  2\begin{pmatrix}
  1 & -0.9 \\[2pt]
  -0.9 & 1
  \end{pmatrix}, \qquad
  \Sigma_{2}^{(Q)} =
  2\begin{pmatrix}
  1 & 0 \\[2pt]
  0 & 1
  \end{pmatrix}.
  \]

  \item[\textbf{(3)}] \textbf{Bottom-left group (mean shift, same covariance).}  
  The two distributions share the same covariance matrix but have shifted means:
  \[
  \boldsymbol{\mu}_{3}^{(P)} = (-2,\,-2)^{\!\top}, \qquad
  \boldsymbol{\mu}_{3}^{(Q)} = (0,\,0)^{\!\top}, \qquad
  \Sigma_{3}^{(P)} = \Sigma_{3}^{(Q)} =
  \begin{pmatrix}
  1 & 0.5 \\[2pt]
  0.5 & 1
  \end{pmatrix}.
  \]
\end{enumerate}

For the high-dimensional simulation, we assume that there exists 2D low-dimensional structure to the high-dimensional densities. Let $d$ be the dimension of the high-dimensional densities, and $V_p$ follows the 3-component normal mixture distribution \eqref{eq:P_mixture}, and $V_q$ follows \eqref{eq:Q_mixture}. To project the low-dimensional $V_p,V_q$ onto $\R^d$, we first construct an orthonormal matrix $K\in \R^{d\times 2}$, whose columns form an orthonormal basis in $\mathbb{R}^{d}$. Specifically, we generate a random Gaussian matrix \(G \in \mathbb{R}^{D \times 2}\) with entries \(G_{ij} \sim \mathcal{N}(0,1)\) and obtain \(K\) from the thin QR decomposition $G = Q R, K = Q$, so that \(K^{\top} K = I_2\).  Then the low-dimensional sample $V_p,V_q$ are projected to $\R^d$ using
\[
X_i = K\, V_i + \varepsilon_i, \qquad
\varepsilon_i \sim \mathcal{N}\!\left(0,\, \sigma^2 I_d\right),
\]
and the resulting samples $X_p,X_q$ are in $\R^d$, and their relative density ratio can be analytically computed as the theoretical truth. We choose to use 2D mixtures as the latent distributions simply because it is easy to visualize the performance.

We include three existing methods in this simulation comparison. 

\begin{enumerate}
    \item \textbf{Unconstrained Least-Squares Importance Fitting \citep[uLSIF]{kanamori2009least}} aims to directly minimize the loss function $\min_{r_\theta}\frac{1}{2}\bE_{Q}\br{\sbr{r_\theta(X)-r(X)}^2}$ over a candidate of class of $r_\theta$ that are expressed as kernel expansions. This method assumes the true ratio function is within the reproducing kernel Hilbert space (RKHS) constructed by the kernel basis, and hence requires a global smooth structure of the true ratio function. We used the python package \textit{densratio} (\url{https://github.com/hoxo-m/densratio_py}) and its implementation of uLSIF, and set the kernel bandwidth to be searched in $(0.1, 0.3, 1.0)$ by cross-validation.
    \item \textbf{Kullback–Leibler Importance Estimation Procedure (KLIEP) \citep[KLIEP]{sugiyama2007direct}} also uses kernel basis to minimize the KL divergence $\min_{r_\theta} \text{D}_{\text{KL}}(p \| r_\theta q)$, and we use the \textit{adapt} python package \citet{de2021adapt} (\url{https://adapt-python.github.io/adapt/}). We use the RBF kernel, with at most 100 kernel centers, and use 5-fold cross-validation to automatically select the optimal kernel bandwidth.
    \item \textbf{The density ratio trick (DR-trick) via binary classification.} This method is discussed in Section~\ref{subsec:connection}. We use logistic regression to estimate the binary classification probability $P(Z=1|X)$ in Section~\ref{subsec:connection}, and use the one-to-one mapping through the Bayes rule to recover the relative density ratio.
\end{enumerate}

We denote our neural network model as MLP in Table~\ref{tb:sim}. For each case shown in Table~\ref{tb:sim}, we sample $n_p=n_q=1000$ equally from $P$ and $Q$ as the training sample, and another 1000 pairs as the testing sample. For each method, we use the training sample to train the relative density ratio function, and evaluate RDR on the testing sample and compute the testing MSE. For the neural network model based on MLP structure, to avoid overfit, we collect another 500 pairs of validation sample and stop the training process early if the validation loss does not improve by $10^{-5}$  over 5 consecutive epochs. The number of samples and all the tuning parameters in every method stay the same throughout all the 5 cases in Table~\ref{tb:sim}. Note that including MLP, all of the methods in comparison can be optimized to achieve better performance for each case; however, our goal here is to use a simple simulation setting to study what might happen when extending these methods in high-dimensional RDR estimation.

\begin{table}[ht!]
\centering
\begin{tabular}{llllll}
\toprule
Dim & Noise & MLP            & uLSIF          & KLIEP          & DR-trick            \\
2   & 0.0   & 0.06 (0.04) & 0.15 (0.03) & \textbf{0.05} (0.01) & 0.20 (0.01) \\
20  & 0.1   & \textbf{0.12} (0.03) & 0.16 (0.04) & 0.16 (0.01) & 0.30 (0.01) \\
40  & 0.1   & \textbf{0.15} (0.02) & 0.20 (0.04) & 0.26 (0.02) & 0.31 (0.01) \\
40  & 0.3   & \textbf{0.17} (0.02) & 0.94 (0.03) & 0.24 (0.02) & 0.29 (0.01) \\
100 & 0.3   & \textbf{0.19} (0.01) & 1.30 (0.02) & 0.31 (0.01) & 0.30 (0.01)\\
\bottomrule
\end{tabular}
\caption{Simulation comparison with existing methods. Each entry represents the mean (sd) of test MSE over 100 replicated studies.}
\label{tb:sim}
\end{table}

Table~\ref{tb:sim} shows that MLP gives the second best performance in the 2D case, and the best performance for all other high-dimensional cases. To further understand the pros and cons of different methods, we provide visualizations in Figure~\ref{fig:sim} for all the 5 cases. The 2D plots are directly based on the testing sample. For the high-dimensional cases, we visualize the RDR values on the low-dimensional testing sample $V_p,V_q$.

Based on Figure~\ref{fig:sim}, the DR-trick cannot generalize well to very high-dimensions, since it tends to give very large or very small binary classifications probabilities to a few more extreme observations. uLSIF is a global kernel method that even in low-dimensions, it cannot correctly reflect the true RDR pattern, and it is most sensitive to the high-dimensional noise perturbation. KLIEP relies on the number of kernel centers; with  the limit of such centers set at 100, as the dimension increases, the fixed number of kernel centers fails to capture the high-dimensional complexity of RDR, and the evaluated RDR tends to be more uniform around 1. For the 4-layer MLP, we also only have a fixed number of neural network parameters as the dimension increases, and inevitably, the performance deteriorates as the dimension increases. But the Hellinger-loss-based MLP is still able to capture the general pattern of the low-dimensional manifold especially in the mean-shift case. The caveat is that, when there is noise perturbation in high-dimensions, MLP tends to polarize the points that are supposed to have equal mass in $P$ and $Q$. This is because the noise perturbation in high-dimensions can shift the empirical support of the two samples and make certain points fall out of the overlapping support.

\begin{figure}[p] 
\centering
\includegraphics[width=\textwidth]{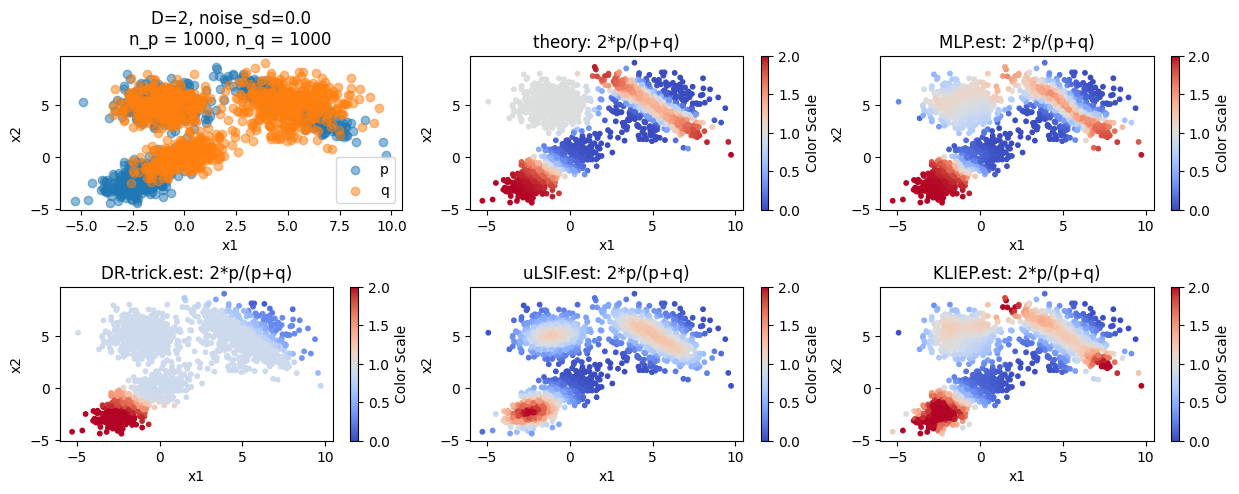}\vspace{3pt}\\
\includegraphics[width=\textwidth]{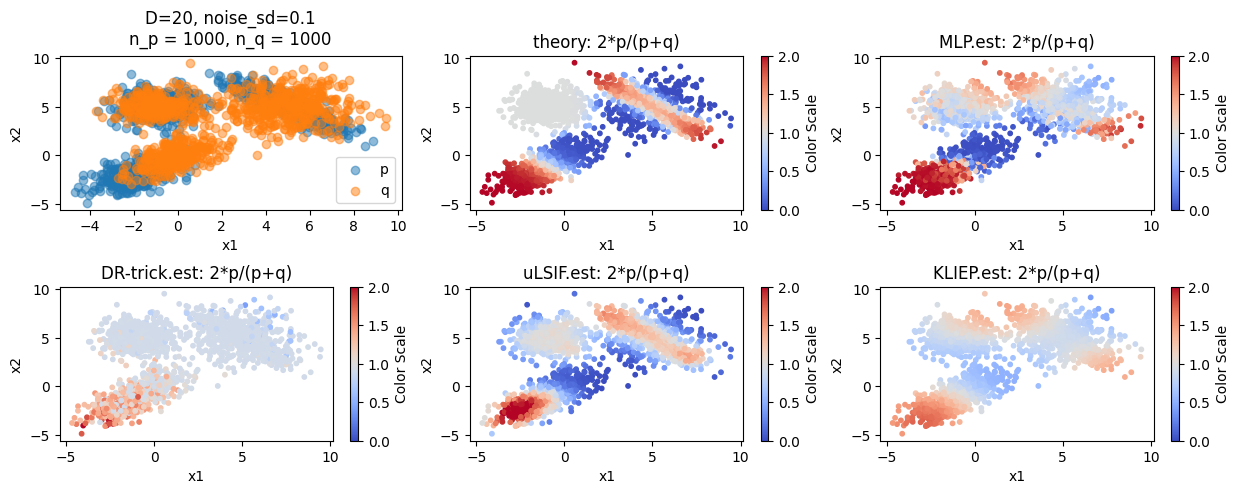}\vspace{3pt}\\
\includegraphics[width=\textwidth]{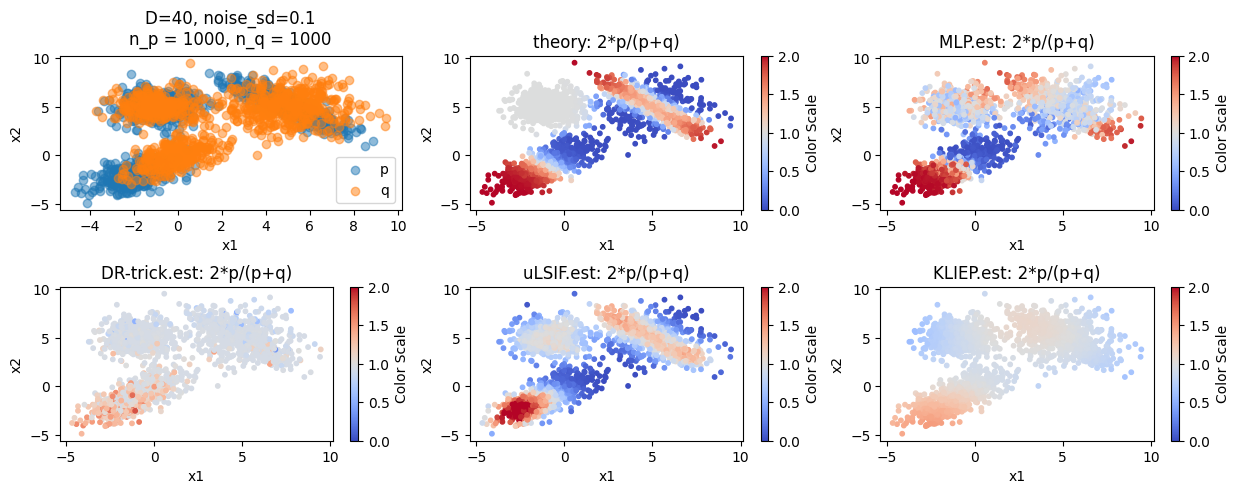}
\caption{Simulation visualizations for each case shown in Table~\ref{tb:sim}.}
\label{fig:sim}
\end{figure}

\begin{figure}[p]\ContinuedFloat
\centering
\includegraphics[width=\textwidth]{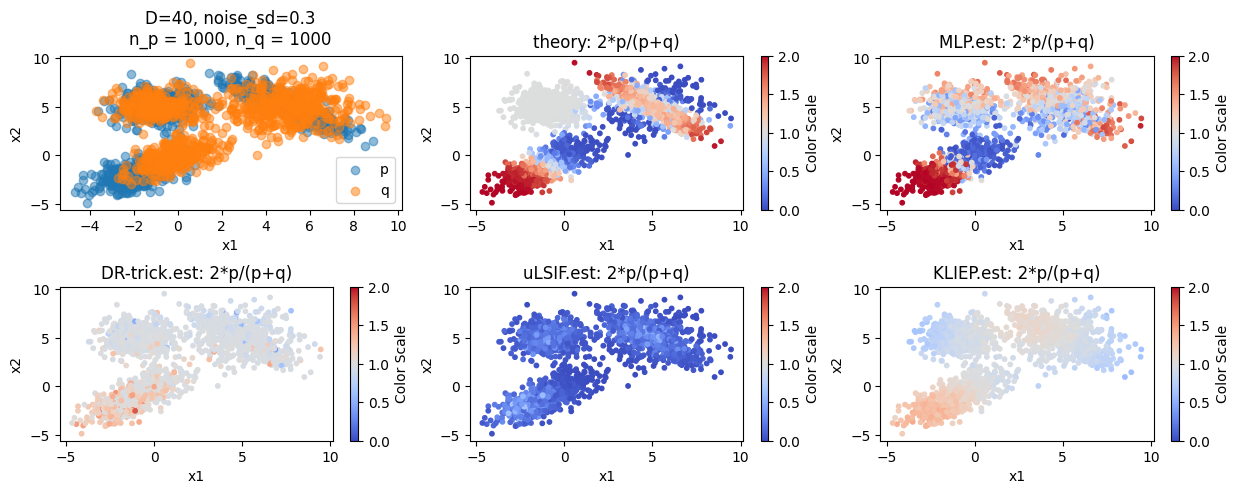}\vspace{3pt}\\
\includegraphics[width=\textwidth]{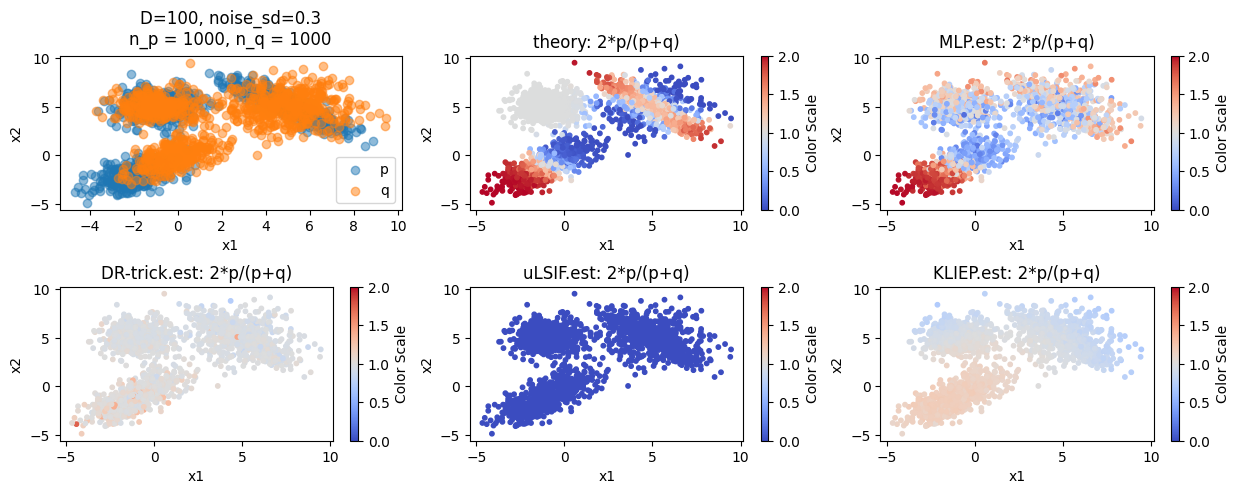}
\caption[]{Simulation visualizations for each case (continued).} 
\end{figure}

\end{document}